\documentclass{llncs}

\usepackage[dvipsnames]{xcolor}
\usepackage{microtype}
\usepackage{amsmath}
\usepackage{amssymb}
\usepackage{stmaryrd}
\usepackage{cite}
\usepackage{paralist}
\usepackage{color}

\input xy
\xyoption{all}

\usepackage{listings}
\usepackage{color}

\usepackage{caption}
\usepackage{subcaption}
\usepackage{hyperref}

\usepackage{algorithm}
\usepackage{algpseudocode}

\usepackage{algorithmicx}

\usepackage{pgf}
\usepackage{tikz}
\usetikzlibrary{arrows,automata}

\newcommand{\ms}[1]{{\bf MS: #1}}

\makeatletter
\newcommand{\xRightarrow}[2][]{\ext@arrow 0359\Rightarrowfill@{#1}{#2}}
\makeatother

\newcommand{\bi}{\begin{array}[t]{@{}l@{}}}
    \newcommand{\ei}{\end{array}}
\newcommand{\ba}{\begin{array}}
    \newcommand{\ea}{\end{array}}
\newcommand{\bda}{\[\ba}
\newcommand{\eda}{\ea\]}
\newcommand{\bp}{\begin{quote}\tt\begin{tabbing}}
        \newcommand{\ep}{\end{tabbing}\end{quote}}

\newcommand{\ignore}[1]{}

\newcommand{\mathem}{\sf}

\newcommand{\rlabel}[1]{\mbox{(#1)}}

\newcommand{\thread}[2]{#1 \sharp #2}

\newcommand{\SSHB}{\mbox{SSHB}}
\newcommand{\SHBPartner}{\mbox{SHBP}}
\newcommand{\HBPartner}{\mbox{HBP}}

\newcommand{\incC}[2]{{\mathem inc}(#1,#2)}

\newcommand{\supVC}[2]{#1 \sqcup #2}

\newcommand{\accVC}[2]{#1[ #2 ]}

\newcommand{\updateVC}[3]{#1[ #2 \mapsto #3 ]}

\newcommand{\pp}{\ \texttt{++}}

\newcommand{\lockE}[1]{\mathit{acq}(#1)} %%{\mathit{acq}(#1)}
\newcommand{\unlockE}[1]{\mathit{rel}(#1)} %%{\mathit{rel}(#1)}
\newcommand{\readE}[1]{r(#1)}
\newcommand{\writeE}[1]{w(#1)}
\newcommand{\accE}[1]{acc(#1)}

\newcommand{\lockET}[1]{\mathit{traceAcq}(#1)} %%{\mathit{acq}(#1)}
\newcommand{\unlockET}[1]{\mathit{traceRel}(#1)} %%{\mathit{rel}(#1)}
\newcommand{\readET}[1]{trace_r(#1)}
\newcommand{\writeET}[1]{trace_w(#1)}

\newcommand{\readEE}[2]{r(#1)_{#2}}
\newcommand{\writeEE}[2]{w(#1)_{#2}}
\newcommand{\lockEE}[2]{acq(#1)_{#2}}  %% acquire
\newcommand{\unlockEE}[2]{rel(#1)_{#2}} %% release

\newcommand{\pos}[1]{\mathit{pos}(#1)}
\newcommand{\posP}[2]{\mathit{pos}_{{\scriptstyle #1}}(#2)}
\newcommand{\compTID}[1]{\mathit{thread}(#1)}
\newcommand{\compTIDP}[2]{\mathit{thread}_{{\scriptstyle #1}}(#2)}

\newcommand{\events}[1]{\mathit{events}(#1)}

\newcommand{\shb}[2]{#1 <^{\scriptscriptstyle SHB} #2}
\newcommand{\strongshb}[2]{#1 <^{\scriptscriptstyle \forall SHB} #2}
\newcommand{\someshb}[2]{#1 <^{\scriptscriptstyle \exists SHB} #2}

\newcommand{\hb}[2]{#1 <^{\scriptscriptstyle HB} #2}
\newcommand{\notHB}[2]{#1 \|^{HB} #2}

\newcommand{\strongSHBRel}{$\forall$SHB}
\newcommand{\someSHBRel}{$\exists$SHB}

\newcommand{\starrow}[1]{\stackrel{#1}{\rightarrow}}

\newcommand{\hbSym}{<^{\scriptscriptstyle HB}}
\newcommand{\strongshbSym}{<^{\scriptscriptstyle \forall SHB}}

\newcommand{\writeVC}[1]{\mathit{Write}(#1)}
\newcommand{\lastWriteVC}[1]{\mathit{L_W}(#1)}

\newcommand{\readVC}[1]{\mathit{Read}(#1)}

\newcommand{\threadVC}[1]{\mathit{Th}(#1)}
\newcommand{\lockVC}[1]{\mathit{Rel}(#1)}

\newcommand{\gtEdge}{\rightarrow} %%{\prec}

\newcommand{\threadLS}[1]{\mathit{Th_{LS}}(#1)}

\newcommand{\rwT}[1]{T^{rw}_{#1}}
\newcommand{\rwTx}{\rwT{x}}

\newcommand{\CountRaces}{\#Races}

\title{Data Race Prediction for Inaccurate Traces}

\author{Martin Sulzmann and Kai Stadtm{\"u}ller}
\institute{
  Faculty of Computer Science and Business Information Systems \\
  Karlsruhe University of Applied Sciences \\
  Moltkestrasse 30, 76133 Karlsruhe, Germany\\
  \email{martin.sulzmann@gmail.com} \\
    \email{kai.stadtmueller@live.de} 
}

\begin{document}

\maketitle

\begin{abstract}
  Happens-before based data race prediction methods infer from a trace of events a partial order
  to check if one event happens before another event.
  If two write events are unordered, they are in a race.
  We observe that common tracing methods provide no guarantee that the
  trace order corresponds to an actual program run.
  We refer to this as an inaccurate trace.
  The consequence  of inaccurate tracing is that results (races) reported are inaccurate.
  We introduce diagnostic methods to examine
  if (1) a race is guaranteed to be correct regardless of 
  any potential inaccuracies, (2) maybe is incorrect due to inaccurate tracing.
  We have fully implemented the approach and provide for an empirical comparison
  with state of the art happens-before based race predictors such as FastTrack and SHB.
\end{abstract}

%%%%%%%%%%%%%%%%%%%%%%%%%%%%%%%%%%%%%%%%%%%%%%%%%%%%%%%%%%%%%%
%%%%%%%%%%%%%%%%%%%%%%%%%%%%%%%%%%%%%%%%%%%%%%%%%%%%%%%%%%%%%%
\section{Introduction}

We consider trace-based dynamic verification methods that
attempt to predict a data race based on a single execution run.
We assume that relevant program events such as as write/read and
acquire/release operations are recorded in a program trace.
This trace represents an interleaved execution of the program
and is the basis for further analysis.

If two conflicting events, e.g.~two writes on the same variable,
appear right next to each other in the trace, a data race may arise
in the actual program execution. Due to the interleaved recording of events,
two conflicting events may not appear right next to
each other in the trace. The challenge is to predict if there is a
valid trace reordering under which both events take place one after the other.

A popular method to meet such challenges is Lamport's happens-before model~\cite{lamport1978time}.
The idea is to infer from the recorded sequence of events a partial order order relation
among events. If two events are ordered then one happens before the other event and
therefore there cannot be any conflict among these events.
If two, say write, events are unordered then there is a potential conflict.
Neither event happens before the other event.
Hence, there is an interleaving where both events take place one after the other.
This may lead to a conflict.
There exists a significant line of research that builds upon the idea
of deriving a happens-before relation from the trace of recorded events,
e.g.~see~\cite{DBLP:journals/concurrency/PoznianskyS07,flanagan2010fasttrack,Smaragdakis:2012:SPR:2103621.2103702,Huang:2015:GGP:2818754.2818856,Kini:2017:DRP:3140587.3062374,Roemer:2018:HUS:3296979.3192385,Mathur:2018:HFR:3288538:3276515,Kalhauge:2018:SDP:3288538.3276516}.

These works rely on the assumption that the trace is \emph{accurate}.
By accurate we mean that operations
can be executed in the same order as their corresponding events appear in the trace.
To derive the sequence of events in a form of a trace, operations need to be instrumented.
In case of a write (read) operation, we record an event to represent
the fact that a write (read) on some shared variable took place.
The same applies to mutex operations such as acquire and release.

The issue is that storing of events is largely unsynchronized to ensure a low run-time overhead.
Hence, the sequence of events recorded in a trace may be \emph{inaccurate} and may not reflect
an actual program run. For happens-before based analyses this means that
release-acquire and write-read happens-before dependencies
derived from the trace are inaccurate. Then,
a predicted race may not be feasible (false positive) and an obvious race may be missed (false negative).
These are scenarios that potentially appear in practice.

Concretely, we consider the happens-before based data race predictors
FastTrack~\cite{flanagan2010fasttrack} and SHB~\cite{Mathur:2018:HFR:3288538:3276515}.
FastTrack is only sound for the first race reported.
FastTrack ignores write-read dependencies and therefore subsequent races reported may be false positives.
SHB on the other hand includes write-read dependencies and comes with the guarantee
that no false positives arise.
This guarantee only holds under the assumption that the trace is accurate.
However, common tracing tools such as
RoadRunner~\cite{flanagan2010roadrunner} may produce inaccurate traces.

Our idea is to provide diagnostic methods to examine the data races reported by
happens-before based race predictors.
We provide methods to check if a race is (1) \emph{guaranteed} to be correct regardless of 
any potential inaccuracies during tracing, (2) \emph{maybe} incorrect due to inaccurate tracing.
Applied to FastTrack and SHB, we can provide the following diagnostic information.
If a FastTrack/SHB race is a guaranteed race we can be certain that
(inaccurate) write-read dependencies have no impact on the analysis result.
If a SHB race is a maybe race, then tracing could be inaccurate.
If a FastTrack race is a maybe race, then this could be a false positive

In summary, our contributions are:
\begin{itemize}
\item We introduce an instrumentation scheme that guarantees accurate tracing
  of release-acquire dependencies (Section~\ref{sec:run-time-events}).
  
\item We develop a family of happens-before relations that take into account
  inaccurately traced write-read dependencies (Section~\ref{sec:happens-before-relation}).

\item We propose diagnostic methods to examine if data races reported are
  affected by inaccurately traced write-read dependencies (Section~\ref{sec:diagnosis}).

\item We have fully implemented these diagnostic methods and provide for an empirical comparison
      with the results provided by FastTrack and SHB (Section~\ref{sec:experiments}).

\end{itemize}

The upcoming section explains by example the issue of inaccurate tracing
and its impact on happens-before based data race prediction methods.
We also motivate our method how to check if the predicted race is affected by inaccurate tracing.
Section~\ref{sec:run-time-events} reviews instrumentation and tracing.
Section~\ref{sec:related-works} covers related work.
Section~\ref{sec:conclusion} concludes.

Our implementation is available via
\begin{center}
  \url{https://github.com/KaiSta/gopherlyzer-GuaranteedRaces}.
\end{center}

The  paper contains an appendix
with further details such as proofs of the results stated
in this paper.

%%%%%%%%%%%%%%%%%%%%%%%%%%%%%%%%%%%%%%%%%%%%%%%%%%%%%%%%%%%%%%
%%%%%%%%%%%%%%%%%%%%%%%%%%%%%%%%%%%%%%%%%%%%%%%%%%%%%%%%%%%%%%
\section{Overview}
\label{sec:overview}

We identify two scenarios where inaccurate tracing
affects the precision of the happens-before relation derived from the trace.

%%%%%%%%%%%%%%%%%%%%%%%%%%%%%%%%%%%%%%%%%%%%%%%%%%%%%%%%%%%%%%
\subsection{Inaccurate Tracing of Release-Acquire Dependencies}

\begin{figure}[tp]
  \centering
\begin{subfigure}{.3\textwidth}
  \centering
\begin{verbatim}
  spawn { // thread 1
   acquire(y);
   x = 1;
   release(y);    
  }
  acquire(y);
  x = 2;
  release(y);
\end{verbatim}
  \caption{Program}
  \label{subfig:program}
\end{subfigure}  
\begin{subfigure}{.3\textwidth}
  \centering
  \bda{ll|l}
   & \thread{1}{}
   & \thread{2}{}
  \\ \hline
  1. & \lockE{y} &
  \\ 2. & \writeE{x} & 
  \\ 3. & \unlockE{y}
  \\ 4. & & \lockE{y}
  \\ 5. & & \writeE{x}
  \\ 6. & & \unlockE{y}
  \eda
  \caption{Actual Trace}
  \label{subfig:actual-trace}
\end{subfigure}  
\begin{subfigure}{.3\textwidth}
  \centering
    \bda{ll|l|l}
   & \thread{1}{}
   & \thread{2}{}
  \\ \hline
  1. & \lockE{y} &
  \\ 2. & \writeE{x} &
  \\ 3. & & \lockE{y}
  \\ 4. & & \writeE{x}
  \\ 5. & \unlockE{y}
  \\ 6. & & \unlockE{y}
  \eda
  \caption{Inaccurate Trace}
    \label{fig:example-lock-false-positive-inaccurate-trace}
\end{subfigure}  
  \caption{Inaccurate Tracing of Release-Acquire: False Positive}
  \label{fig:example-lock-false-positive}
\end{figure}

We consider the program in Figure~\ref{fig:example-lock-false-positive} (on the left).
We assume a mutex referenced by variable~$y$ and some shared variable~$x$.
In the main thread, referred to as thread~2, we acquire the mutex,
write to $x$, and then release the mutex.
The main thread creates another thread, referred to as thread~1,
which effectively performs the same sequence of operations.
We assume that the program is instrumented such that the execution of the program
yields a trace. A trace is a sequence of events that took place and
is meant to represent the interleaved execution of the various threads found in the program.
The events we are interested in here are acquire, release, read and write.

Subfigure~\ref{subfig:actual-trace} shows such a trace.
We use a tabular notation to record for each event, the kind of event,
its position in the trace and the thread id of the thread in which the event took place.
For example, an acquire event on $y$ in thread~1 is recorded at trace position~1.
We use notation $\thread{1}{\lockEE{y}{1}}$ to represent this information.
We write $\thread{2}{\writeEE{x}{5}}$ to denote the write event on $x$ in thread~2 recorded
at trace position~5.
We often skip the thread identifier and write $\writeEE{x}{5}$ for short.

A well-explored data race prediction approach is to infer from the trace
a happens-before relation. If two conflicting events, e.g.~two writes on the same variable,
are not in a happens-before (HB) relation we argue that the trace can be reordered
such that both events appear next to each other in the trace.
This shows that there is a potential data race.

For the trace in Subfigure~\ref{subfig:actual-trace} we infer that
(a) $\writeEE{x}{2} < \unlockEE{y}{3}$,
(b) $\unlockEE{y}{3} < \lockEE{y}{4}$, and
(c) $\lockEE{y}{4} < \writeEE{x}{5}$.
HB relations (a) and (c) are due to the program order.
That is, for each thread events take place in the order as they appear in the program text.
HB relation (b) is due to a release-acquire dependency.
An acquire can only be performed once the mutex is released. Hence, we impose (b).
Based on the above HB order, we find that $\writeEE{x}{2} < \writeEE{x}{5}$.
We conclude that the two writes are not in a race.

The issue we investigate in this work is that tracing may be inaccurate.
The program text is instrumented to record an event
for each operation we are interested in.
To ensure a low run-time overhead,
there is almost no synchronization among concurrently taking place recordings.
This may lead to some inaccurate trace.

The example in Figure~\ref{fig:example-lock-false-positive}
has two threads of execution.
The release operation in thread~1 takes place before the acquire
and write operation in thread~2.
Assuming that all events are recorded after the respective operation took place,
it is entirely possible that the recording of the release
event will be overtaken by the acquire and write event.
This then results in the (inaccurate) 
trace reported in Subfigure~\ref{fig:example-lock-false-positive-inaccurate-trace}.

Like above, we find the program order relations (a) and (c).
However, the release-acquire dependency is lost.
In the recorded trace in Subfigure~\ref{fig:example-lock-false-positive-inaccurate-trace},
the release in thread~1 does not precede the acquire in thread~2.
The consequence is that the two writes appear unordered and the analysis (falsely)
signals that there is a potential race.

The above scenario may arise in practice. The RoadRunner tracing tool~\cite{flanagan2010roadrunner}
possibly yields the inaccurate trace in Subfigure~\ref{fig:example-lock-false-positive-inaccurate-trace}.
A happens-before based data race predictor
such as FastTrack~\cite{flanagan2010fasttrack} will then (wrongly) signal a write-write race.
This is clearly a false positive as for any execution run the writes are mutually exclusive.

%%%%%%%%%%%%%%%%%%%%%%%%%%%%%%%%%%%%%%%%%%%%%%%%%%%%%%%%%%%%%%
\subsection{Inaccurate Tracing of Write-Read Dependencies}

Another source of inaccurate tracing are unsynchronized reads and writes.
Write-read dependencies derived from the trace may not be accurate
and we may then encounter false positives and negatives.

\begin{figure}[tp]
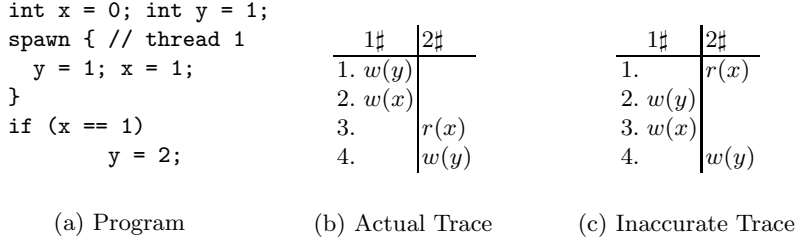

  \centering
\begin{subfigure}{.3\textwidth}
  \centering
\begin{verbatim}
  int x = 0; int y = 1;
  spawn { // thread 1
    y = 1; x = 1; 
  }
  if (x == 1)
          y = 2;
\end{verbatim}
  \caption{Program}
\end{subfigure}  
\begin{subfigure}{.3\textwidth}
  \centering
  \bda{ll|l}
   & \thread{1}{}
   & \thread{2}{}
  \\ \hline
  1. & \writeE{y} &
  \\ 2. & \writeE{x} &
  \\ 3. & & \readE{x}
  \\ 4. & & \writeE{y}
  \eda
  \caption{Actual Trace}
  \label{fig:example-wrd-false-positive-actual}
\end{subfigure}  
\begin{subfigure}{.3\textwidth}
  \centering
    \bda{ll|l}
   & \thread{1}{}
   & \thread{2}{}
  \\ \hline
  1. & & \readE{x}
  \\ 2. & \writeE{y} &
  \\ 3. & \writeE{x} &
  \\ 4. & & \writeE{y}
  \eda
  \caption{Inaccurate Trace}
  \label{fig:example-wrd-false-positive-inaccurate}
\end{subfigure}  
  \caption{Inaccurate Tracing of Write-Read Dependencies: False Positive}
  \label{fig:example-wrd-false-positive}
\end{figure}

Consider the program in Figure~\ref{fig:example-wrd-false-positive} (on the left).
We find two threads where the main thread is referred to as thread~2.
We assume shared variables $x$ and $y$. We include their declaration and initial values (zero).
This did not matter for our earlier example.
We run the program where we assume that thread~1 executes first.
The resulting trace is shown in Subfigure~\ref{fig:example-wrd-false-positive-actual}.
For brevity, the trace does not contain the initial writes to \texttt{x} and \texttt{y}.

For the trace in Subfigure~\ref{fig:example-wrd-false-positive-actual} we infer that
(a) $\writeEE{y}{1} < \writeEE{x}{2}$,
(b) $\readEE{x}{3} < \writeEE{y}{4}$, and
(c) $\writeEE{x}{2} < \readEE{x}{3}$.
HB relations (a) and (b) are due to the program order. 
HB relation (c) is due to a write-read dependency.
As the control flow may be affected by the read value,
we assume that the read event happens after the preceding write.
Hence, we impose (c). Based on the above HB order we then find that $\writeEE{y}{1} < \writeEE{y}{4}$.
We conclude that the two writes on \texttt{y} are not in a race.

We take a closer look at the recording of events when running the program.
After each read/write operation, we record a read/write event.
We consider a specific execution run where
thread~1 executes first followed by the if-statement in the main thread.
In detail, after execution of \texttt{y = 1} and \texttt{x = 1} in thread~1,
we execute the if-condition \texttt{x == 1} followed by \texttt{y = 2} in the main thread (2).
For the resulting events we observe the following.

Thread~1 records the events $\writeE{y}$ and $\writeE{x}$
and the main thread records the events $\readE{x}$ and $\writeE{y}$.
For each thread, events are recorded in the sequence as they arise.
For example, for thread~1, we find that $\writeE{y}$ appears before $\writeE{x}$
in the trace. Recording of concurrent events is not synchronized.
Hence, it is possible that the recording of $\readE{x}$ takes place
before the recording of $\writeE{x}$.
Then, we encounter the trace in Subfigure~\ref{fig:example-wrd-false-positive-inaccurate}.

This trace implies the same program order relations (a) and (b) as
discovered for the trace in Subfigure~\ref{fig:example-wrd-false-positive-actual}.
However, the write-read dependency (c) is lost because $\writeEE{x}{2}$
no longer precedes $\readEE{x}{3}$.
Then, the two write events on $y$ are unordered
and the analysis (falsely) signals that there is a potential race.

The above scenario may arise in practice.
A happens-before based data race predictor such as SHB~\cite{Mathur:2018:HFR:3288538:3276515}
that imposes write-read dependencies then reports a race where where there is actually none.
The initial value of $x$ is zero. The then-statement will only be reached once thread~1 is fully executed.
Hence, the two writes on $y$ always happen one after the other. Hence, there is no write-write race on $y$.

\begin{figure}[tp]
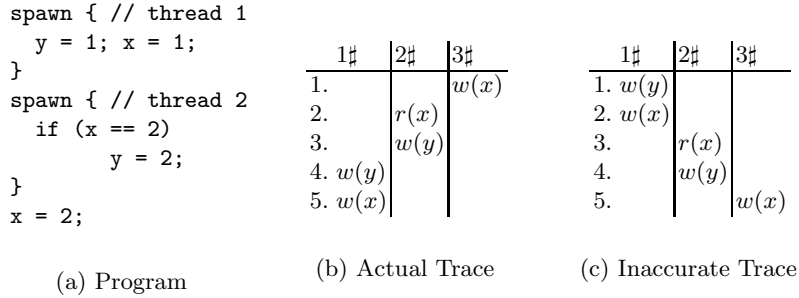

  \centering
\begin{subfigure}{.3\textwidth}
  \centering
\begin{verbatim}
  spawn { // thread 1
    y = 1; x = 1; 
  }
  spawn { // thread 2
    if (x == 2)
          y = 2;
  }
  x = 2;
\end{verbatim}
  \caption{Program}
\end{subfigure}  
\begin{subfigure}{.3\textwidth}
  \centering
  \bda{ll|l|l}
   & \thread{1}{}
   & \thread{2}{}
   & \thread{3}{}
  \\ \hline
  1. & & & \writeE{x}
  \\ 2. & & \readE{x} &
  \\ 3. & & \writeE{y} &
  \\ 4. & \writeE{y} & &
  \\ 5. & \writeE{x} & &
  \eda
  \caption{Actual Trace}
    \label{fig:example-wrd-false-negative-actual-trace}  
\end{subfigure}  
\begin{subfigure}{.3\textwidth}
  \centering
    \bda{ll|l|l}
   & \thread{1}{}
   & \thread{2}{}
   & \thread{3}{}
  \\ \hline
  1. & \writeE{y} & &
  \\ 2. &  \writeE{x} & &
  \\ 3. & & \readE{x} &
  \\ 4. & & \writeE{y} &
  \\ 5. & & & \writeE{x}
  \eda
  \caption{Inaccurate Trace}
    \label{fig:example-wrd-false-negative-inaccurate-trace}
\end{subfigure}  
  \caption{Inaccurate Tracing of Write-Read Dependencies: False Negative}
  \label{fig:example-wrd-false-negative}
\end{figure}

Figure~\ref{fig:example-wrd-false-negative} shows another example.
For this case, inaccurate tracing of write-read dependencies omits a race.
For the trace in Subfigure~\ref{fig:example-wrd-false-negative-inaccurate-trace} we infer that
(a) $\writeEE{y}{1} < \writeEE{x}{2}$,
(b) $\writeEE{x}{2} < \readEE{x}{3}$,
(c) $\readEE{x}{3} < \writeEE{y}{4}$.
HB relations (a) and (c) are due to the program order.
Relation (b) is due to a write-read dependency.
We conclude that $\writeEE{y}{1}$ happens-before $\writeEE{y}{4}$.
Hence, there is no race.
This reasoning is wrong due to inaccurately traced events.
Based on the actual (accurate) trace in Subfigure~\ref{fig:example-wrd-false-negative-actual-trace},
we find that the two writes on $y$ are in a race.

%%%%%%%%%%%%%%%%%%%%%%%%%%%%%%%%%%%%%%%%%%%%%%%%%%%%%%%%%%%%%%
\subsection{Our Approach}

Inaccurate tracing of release-acquire and write-read dependencies
affects the precision of happens-before based data race predictors.
 See the above examples.
In case of release-acquire, we will propose an alternative (yet simple) instrumentation scheme
that by construction guarantees that all release/acquire dependencies are accurately traced.
No such guarantees can be given for write-read dependencies unless some heavy-weight instrumentation
scheme is used that severely impacts the run-time performance.
Hence, we have to face the possibility that write-read dependencies
derived from the trace are not accurate.

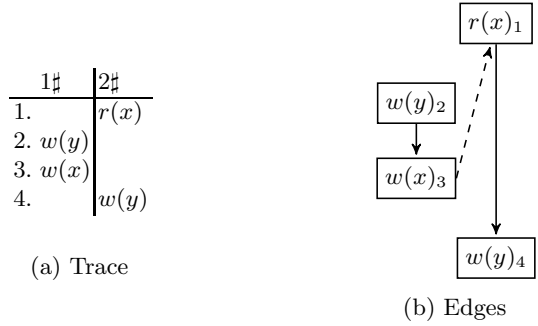
\begin{figure}[tp]
  \centering

\begin{subfigure}{.4\textwidth}
  \centering
    \bda{ll|l}
   & \thread{1}{}
   & \thread{2}{}
  \\ \hline
  1. & & \readE{x}
  \\ 2. & \writeE{y} &
  \\ 3. & \writeE{x} &
  \\ 4. & & \writeE{y}
  \eda
  \caption{Trace}
\end{subfigure}  
\begin{subfigure}{.4\textwidth}
  \centering

\begin{tikzpicture}[->,>=stealth',shorten >=1pt,auto,node distance=1.5cm,semithick]
            \node[draw=black] (A) {$\readEE{x}{1}$};
            \node[draw=black] (B) [below left of=A] {$\writeEE{y}{2}$};
            \node[draw=black] (C) [below of=B, yshift=.5cm] {$\writeEE{x}{3}$};
            \node[draw=black] (D) [below right of=C] {$\writeEE{y}{4}$};
            
            \path (B) edge node {} (C);
            \path (C.east) edge [dashed] node  {} (A);            
            \path (A) edge node {} (D);
  \end{tikzpicture}
%% MS: omit, use tikz instead  
%% \xymatrix{
%%   & \readEE{x}{1} \ar[ddd]
%%   \\
%%   \writeEE{y}{2} \ar[d] &
%%   \\
%%   \writeEE{x}{3} \ar@{..>}[uur] &
%%   \\
%%    & \writeEE{y}{4}
%%   }
 \caption{Edges}
\end{subfigure}
\caption{Data Race Diagnosis for Subfigure~\ref{fig:example-wrd-false-positive-inaccurate}}
\label{fig:diagnosis-a}
\end{figure}

Our approach to deal with inaccurate write-read dependencies is as follows.
The standard write-read dependency relation
imposes a happens-before relation from the \emph{nearest write} in the trace to the read.
We take into account inaccurate tracing by considering
\emph{all potential write} dependency candidates.
A write is a potential write-read dependency candidate for a read
if (1) the read and write are unsynchronized, or
(2) the write immediately happens before the read with no other write in between.

We then build a graph to examine pairs of events that are in a race.
For each standard happens-before relation we add an edge.
In addition, we add an edge from a write to a read if the write is a potential
write-read dependency candidate for that read.
Thus, we take into account inaccuracies that we might experience during
tracing of unsynchronized reads/writes.
For each data race pair we then test if there is a path that connects both events.
If there is a path, then we classify the pair as a \emph{maybe} data race.
If there is no path, then we classify the pair is a \emph{guaranteed} data race.

We revisit our earlier examples.
Figure~\ref{fig:diagnosis-a} shows the resulting edges (on the right)
for the (inaccurate) trace in Subfigure~\ref{fig:example-wrd-false-positive-inaccurate}.
Solid edges correspond to standard happens-before relations.
Dotted edges correspond to potential write-read dependencies.
There is only one write candidate for the read here.
The standard happens-before analysis reports for the pair $(\writeEE{y}{2}, \writeEE{y}{4})$ that the events are unordered.
There is a path from $\writeEE{y}{2}$ to $\writeEE{y}{4}$.
Hence, this is a \emph{maybe data race} because there is a possible choice of write-read dependency
under which both writes are not in a race.

\begin{figure}[tp]
  \centering
\begin{subfigure}{.4\textwidth}
  \centering
    \bda{ll|l|l}
   & \thread{1}{}
   & \thread{2}{}
   & \thread{3}{}
  \\ \hline
  1. & \writeE{y} & &
  \\ 2. &  \writeE{x} & &
  \\ 3. & & \readE{x} &
  \\ 4. & & \writeE{y} &
  \\ 5. & & & \writeE{x}
  \eda
  \caption{Trace}
\end{subfigure}
\begin{subfigure}{.4\textwidth}
  \centering
\begin{tikzpicture}[->,>=stealth',shorten >=1pt,auto,node distance=1.5cm,semithick]
            \node[draw=black] (A) {$\writeEE{y}{1}$};
            \node[draw=black] (B) [below of=A, yshift=.5cm] {$\writeEE{x}{2}$};
            \node[draw=black] (C) [below right of=B] {$\readEE{x}{3}$};
            \node[draw=black] (D) [below of=C, yshift=.5cm] {$\writeEE{y}{4}$};
            \node[draw=black] (E) [below right of=D] {$\writeEE{x}{5}$};            
            
            \path (A) edge node {} (B);
            \path (B) edge [dashed] node  {} (C);            
            \path (C) edge node {} (D);
            \path (E.north) edge [dashed, bend right=15] node {} (C.east);            
  \end{tikzpicture}  
%% MS: omit, use tkiz instead
%% \xymatrix{
%%   \writeEE{y}{1} \ar[d] & & 
%%   \\ \writeEE{x}{2} \ar@{..>}[dr] & &
%%   \\  & \readEE{x}{3} \ar[d] &
%%   \\ & \writeEE{y}{4} &
%%   \\ & & \writeEE{x}{5} \ar@{..>}[uul]
%%  }
   \caption{Edges}
\end{subfigure}  
\caption{Data Race Diagnosis for Subfigure~\ref{fig:example-wrd-false-negative-inaccurate-trace}}
\label{fig:diagnosis-b}
\end{figure}
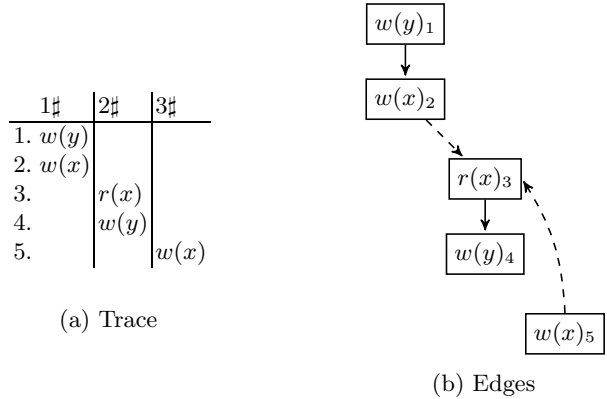

Figure~\ref{fig:diagnosis-b} shows
the (inaccurate) trace from Subfigure~\ref{fig:example-wrd-false-negative-inaccurate-trace}.
Edges are on the right. There are two write candidates for the read that could possibly
form a write-read dependency. We are conservative and include all possibilities.
Events $\writeEE{x}{2}$ and $\writeEE{x}{5}$ are unordered under the standard happens-before relation (considering solid edges only).
Considering the entire graph, solid and dotted edges, we find that there is no path that connects both events.
Hence, the pair $(\writeEE{x}{2}, \writeEE{x}{5})$ is a \emph{guaranteed data race}.
Events $\writeEE{y}{1}$ and $\writeEE{y}{4}$ are also unordered but there is a path that connects both events.
Hence, this pair represents a maybe data race.

Our diagnostic method to identify maybe and guaranteed data races requires two phases.
The first phase resembles existing algorithms such as FastTrack and SHB.
We achieve the same (worst-case) time complexity but require (worst-case) quadratic, in the size of the trace, space
to store edges. We therefore apply our method always offline, unlike FastTrack and SHB that can applied online
while the program is running.
The second phase then examines data races reported in the first phase. For each such data race pair,
we perform a graph search that requires (worst-case) quadratic, in the size of the trace, time.

We have fully implemented our approach. Benchmarks for some real-world examples show
that our method is practical and is useful to examine the results reported
by FastTrack~\cite{flanagan2010fasttrack} and SHB~\cite{Mathur:2018:HFR:3288538:3276515}.
FastTrack ignores write-read dependencies whereas SHB imposes write-read dependencies.
Hence, FastTrack is only sound for the first race reported.

If a FastTrack/SHB race is a guaranteed race we can be certain that
(inaccurate) write-read dependencies have no impact on the analysis result.
If a SHB race is a maybe race, then tracing could be inaccurate.
If a FastTrack race is a maybe race, then this could be a false positive
because FastTrack ignores write-read dependencies.
The details of our approach follow below.

%%%%%%%%%%%%%%%%%%%%%%%%%%%%%%%%%%%%%%%%%%%%%%%%%%%%%%%%%%%%%%
%%%%%%%%%%%%%%%%%%%%%%%%%%%%%%%%%%%%%%%%%%%%%%%%%%%%%%%%%%%%%%
\section{Instrumentation and Tracing}
\label{sec:run-time-events}

We assume concurrent programs making use of shared variables
and acquire/release (a.k.a.~lock/unlock) primitives.
Further constructs such as fork and join are omitted for brevity.
Programs will be instrumented to derive a trace of events
when running the program.
The information recorded in the trace is of the following form.

\begin{definition}[Run-Time Traces and Events]
\label{def:run-time-traces-events}  
\bda{lcll}
  T & ::= & [] \mid \thread{i}{e} : T   & \mbox{Trace}
  \\ e,f,g & ::= &  \readEE{x}{j}
           \mid \writeEE{x}{j}
           \mid \lockEE{y}{j}
           \mid \unlockEE{y}{j}

           & \mbox{Events}
\eda
\end{definition}

A trace $T$ is a list of events. We adopt Haskell notation for lists
and assume that the list of objects $[o_1,\dots,o_n]$ is a shorthand
for $o_1:\dots:o_n:[]$. We write $\pp$ to denote the concatenation operator among lists.
For each event $e$, we record the thread id number $i$ in which the event took place,
written $\thread{i}{e}$.
We write $\readEE{x}{j}$ and $\writeEE{x}{j}$
to denote a read and write event on shared variable $x$ at position $j$.
We write $\lockEE{y}{j}$ and $\unlockEE{y}{j}$
to denote a lock and unlock event on mutex $y$
The number $j$ represents the position of the event in the trace.
We sometimes omit the thread id and position for brevity.

To instrument programs we assume a syntactic pre-processing step
where the source program is manipulated as follows.

\begin{definition}[Instrumentation]
\label{def:instrumentation}  
  \bda{cc}

  \rlabel{Acq} \
  \ba{lcl}
    \texttt{acquire(m);} & \Rightarrow &
  \ba{l}
  \texttt{acquire(m);}
  \\   \texttt{TRACE\_ACQ(m);}
  \ea
  \ea
  \\
  \rlabel{Rel} \
  \ba{lcl}
  \texttt{release(m);} & \Rightarrow &
  \ba{l}
  \texttt{TRACE\_REL(m);}
  \\ \texttt{release(m);}  
  \ea  
  \ea
  &
  \rlabel{RW} \
    \ba{lcl}
  \texttt{x = y;} & \Rightarrow &
  \ba{l}
  \texttt{tmp = y;}
  \\ \texttt{TRACE\_READ(y);}
  \\ \texttt{x = tmp;}
  \\ \texttt{TRACE\_WRITE(y);}
  \ea
  \ea
  \eda
\end{definition}  
We assume some (trace recording) primitives \texttt{TRACE\_X} that store the respective events in some trace.
For example, \texttt{TRACE\_ACQ(m)} issues the event $\thread{i}{\lockEE{m}{k}}$
and stores this event at trace position~$k$.
For brevity, we keep symbolic names such as~$m$. In the actual instrumentation,
they will be represented via hash codes.
A further technical detail is that the trace recording primitives
must have access to the thread id~$i$ associated to this program text.
We also need to increment the current trace position~$k$.

To avoid conflicts when generating and storing events,
each \texttt{TRACE\_X} operation is executed atomically.
We assume a program execution model where an atomically executed
operation issues a memory barrier. 
A consequence of introducing these memory barriers is that
the above instrumentation scheme enforces a sequential consistency memory model~\cite{Adve:1996:SMC:619013.620590}.
That is, for each thread all reads and writes are executed in the sequence
as they appear in the program text.

We consider the instrumentation of acquire/release operations.
We assume that for each acquire there is a matching release operation in the same thread.
This is a common requirement that must be fulfilled by mutex operations.
Each acquire issues the event \emph{after} the operation has executed.
In case of release, we issue the event \emph{before} the release operation executes.
This is done for the following reason.

After the acquire operation is executed no other acquire event for the same mutex
will be issued (due to mutual exclusion). We first record the release event before
we actually perform the release operation. Hence, the release event is issued right after the acquire event
and there can no be other competing release events on the same mutex.
The matching release must be in the same thread. Hence, the acquire and release event
are issued in sequence and this sequence is preserved when storing events.

In case of a prematurely terminated program where some acquire lacks a matching release,
we introduce a dummy release event. We conclude that the
instrumentation and tracing scheme laid out above, enjoys the following guarantees:
\begin{description}
\item[INST\_PO:] Events in a thread
  are traced in the order as they appear in the program text.
  
 \item[INST\_RA:] Release and acquire events are traced in the order as they
   are executed and therefore each acquire can be properly matched with its corresponding release
   by scanning through the trace.
\end{description}

We consider in more detail the instrumentation of reads and writes.
For reads and writes we issue a trace recording primitive \emph{after} the actual operation took place.
To ensure a low run-time overhead and avoid blocking of the running thread,
trace recording relies on lock-free operations and/or perform
the trace recording in a separate helper thread.
The issue is that intra-thread event recordings may be out of sync
relative to the actual program execution order.
Hence, it is possible that the event of an earlier read appears after some write event
if both event recordings take place concurrently.
Recording the read/write event \emph{before} the actual operation does not help either.

We would need to synchronize the actual operation
\emph{and} the recording of the associated event by carrying out both steps atomically.
Under a pessimistic scheme, shared operations on non-conflicting variables would
need to be synchronized as well. This is clearly not practical.
The work reported in~\cite{Cao:2017:HRD:3134419.3108138} explores some more optimistic schemes.
A managed run-time is assumed to track dependencies.
The run-time overhead can still be significant as the experimental results show.

Hence, we assume a much more light-weight instrumentation scheme
where there is no synchronization among concurrently taking place tracing operations.
Such a scheme for example is applied by the popular
RoadRunner tracing tool~\cite{flanagan2010roadrunner}.
The down-side is that there is no guarantee that read/write events that
are not protected by a mutex are accurately ordered in the trace.
To deal with such inaccurate traces we apply
some diagnostic methods to distinguish
between guaranteed and maybe data races.

%%%%%%%%%%%%%%%%%%%%%%%%%%%%%%%%%%%%%%%%%%%%%%%%%%%%%%%%%%%%%%
%%%%%%%%%%%%%%%%%%%%%%%%%%%%%%%%%%%%%%%%%%%%%%%%%%%%%%%%%%%%%%
\section{Happens-Before Relation}
\label{sec:happens-before-relation}

We revisit the standard happens-before (HB) relation
and the schedulable happens-before (SHB) relation
that incorporates write-read dependencies.
SHB determines write-read dependencies based on the position in the trace.
We assume that tracing of unsynchronized reads and writes is inaccurate
and therefore some write-read dependencies are missing or wrongly imposed.
We introduce more general SHB relations that take into account such inaccuracies.
This provides the theoretical basis for our diagnostic methods in the upcoming section.

First, we introduce some notation and helper functions.
For trace $T$, we assume some functions to access the thread id and position of $e$. 
We define $\compTIDP{T}{e} = j$ if $T=T_1 \pp\ [\thread{j}{e}] \pp\ T_2$ for some traces $T_1, T_2$.
We define $\posP{T}{\readEE{x}{j}} = j$,
$\posP{T}{\writeEE{x}{j}} = j$,
$\posP{T}{\lockEE{y}{j}} = j$
and
$\posP{T}{\unlockEE{y}{j}} = j$
to extract the trace position from an event.
We assume that the trace position is \emph{correct}:
If $\posP{T}{e} = n$ then $T=\thread{i_1}{e_1}: \dots : \thread{i_{n-1}}{e_{n-1}} : \thread{i}{e} : T'$
for some events $\thread{i_k}{e_k}$ and trace $T'$.
We sometimes drop the component $T$ and
write $\compTID{e}$ and $\pos{e}$ for short.

%% MS: omit, not needed
%% Given a trace $T$, we can also access an event at a certain position $k$.
%% We define $\accTr{T}{k} = e$ if $e \in T$ where $\posP{T}{e} = k$.

For trace $T$, we define $\events{T} = \{ e \mid \exists T_1,T_2,j. T = T_1 \pp [\thread{j}{e}] \pp T_2 \}$
to be the set of events in $T$.
We write $e \in T$ if $e \in \events{T}$.

We consider the standard happens-before (ordering) relation among events.

\begin{definition}[Happens-Before]
\label{def:happens-before}  
  Let $T$ be a trace.
  We define a relation $\hb{}{}$ among trace events
  as the smallest partial order such that the following holds:
  \begin{description}
  \item[Program order (PO):]
    Let $e, f \in T$. Then, $\hb{e}{f}$ iff
    $\compTID{e} = \compTID{f}$ and $\pos{e} < \pos{f}$.
  \item[Release-acquire dependency (RAD):]
    Let $\unlockEE{y}{j}, \lockEE{y}{k} \in T$. Then,
    we define $\hb{\unlockEE{y}{j}}{\lockEE{y}{k}}$ iff
    $j < k$  where
    $\compTID{\unlockEE{y}{j}} \not= \compTID{\lockEE{y}{k}}$ and
    for all $e \in T$ where $j < \pos{e}$, $\pos{e} < k$
    and $\compTID{\unlockEE{y}{j}} \not= \compTID{e}$ we find that $e$
    is not an acquire event on $y$.
  \end{description}
  We refer to $\hb{T}{}{}$ as the \emph{happens-before} relation derived from trace $T$.

    We say two events $e, f \in T$ are \emph{unsynchronized} with respect to each other
  if neither $\hb{e}{f}$, nor $\hb{f}{e}$ holds.
  In such a situation, we write $\notHB{e}{f}$.
\end{definition}

The guarantees {\bf INST\_PO} and {\bf INST\_RA} of our instrumentation and tracing scheme ensure
that the happens-before relation derived from a trace
is not affected by any inaccuracies during tracing.
Hence, for two events $e, f$ in happens-before relation $\hb{e}{f}$ we guarantee
that $e$ was executed before~$f$.

Mathur and coworkers~\cite{Mathur:2018:HFR:3288538:3276515}
extend the happens-before relation by imposing the following condition.

\begin{definition}[Schedulable Happens-Before~\cite{Mathur:2018:HFR:3288538:3276515}]
\label{def:schedulable-happens-before}  
  Let $T$ be a trace.
  We define a relation $\shb{}{}$ among trace events
  as the smallest partial order that satisfies
  PO and RAD from Definition~\ref{def:happens-before}
  and the following additional condition:
  \begin{description}
  \item[Write-read dependency (WRD):]
    Let $\readEE{x}{j}, \writeEE{x}{k} \in T$.
    Then, $\shb{\writeEE{x}{j}}{\readEE{x}{k}}$ iff
    $j < k$ %% $\pos{\writeEE{x}{j}} < \pos{\readEE{x}{k}}$
    and
    for all $e\in T$ where
    $j < \pos{e}$ %%$ \pos{\writeEE{x}{j}} < \pos{e}$
    and $\pos{e} < k$ %% $\pos{e} < \pos{\readEE{x}{k}}$
    we find that $e$ is not a write event on $x$.
  \end{description}
  We refer to $\shb{}{}$ as the \emph{schedulable happens-before}
  relation derived from trace $T$.
\end{definition}
Condition WRD states that if a write precedes a read and there is no other write in between,
then the read must happen after that write.
The intention is that write-read dependencies affecting the program execution
must be respected when analyzing the trace.
Without the WRD condition, we may wrongly conclude that two events are unsynchronized.
This potentially leads to false positives in case of happens-before based data race prediction.

Recall the program in Figure~\ref{subfig:program} and
the trace in Figure~\ref{subfig:actual-trace} obtained from running the instrumented program.
The read value at trace position~$2$ depends on the prior write at trace position~$1$.
The control flow is affected by the read value. The if-condition holds and therefore \texttt{y = 2}
is executed. This results in the write event at trace position~$3$.
Based on the schedulable happens-before relation, we conclude that the write at trace position~$1$
happens before the write at trace position~$3$.
In detail, we find $\shb{\writeEE{x}{1}}{\readEE{x}{2}}$ due to the WRD condition
and $\shb{\readEE{x}{2}}{\writeEE{x}{3}}$ due to the PO condition.
Hence, $\shb{\writeEE{x}{1}}{\writeEE{x}{3}}$.

Without the WRD condition, $\writeEE{x}{1}$ and $\writeEE{x}{3}$ appear to be unsynchronized.
Hence, we predict that there must be
an interleaved execution where both events appear right next to each other in the resulting trace.
This prediction is clearly wrong as $\writeEE{x}{1}$ must happen before $\writeEE{x}{3}$
by observing the possible control flow when executing the program.

We conclude that the WRD condition is important to make valid predictions about the run-time behavior of programs.
The WRD condition in Definition~\ref{def:schedulable-happens-before}
identifies the write that affects the read based on the position
in the trace. In the presence of inaccurate tracing,
the event of the actual \emph{last write} that affects the read might not be
the closest write event in terms of the position in the trace.
Recall the examples in Section~\ref{sec:overview}.
Hence, the schedulable happens-before relation derived from the trace may suggest
trace reorderings that do not correspond to any actual (interleaved) program execution.
Hence, the WRD condition is sensitive to the accuracy of tracing.

To make valid predictions in the presence of inaccurate tracing,
we generalize the WRD condition as follows.
Instead of characterizing the last write that precedes the read in terms
of the position in the trace, we characterize all potential WRD candidates.
Potential WRD candidates are either
(1) writes that are unsynchronized with respect to the read, or
(2) writes that immediately happen before that read.
To identify cases (1) and (2) we rely on the standard happens-before relation
$\hb{}{}$. Recall that $\hb{}{}$ is not affected by inaccurate tracing
of unsynchronized reads and writes.

\begin{definition}[Read/Write Events]
  Let $T$ be a well-formed trace.
  We define $\rwTx$ as the set of all read/write
  events in $T$ on some variable $x$.
\end{definition}

\begin{definition}[WRD Candidates]
\label{def:wrd-candidates}  
Let $T$ be a trace and $\readEE{x}{j} \in T$.
We define
\bda{lcl}
W_1 & = & \{ \writeEE{x}{i} \mid \writeEE{x}{i} \in T \wedge \notHB{\readEE{x}{j}}{\writeEE{x}{i}} \wedge
 \\ && \ \ \ \ \ \ \ \ \ \ 
  \neg \exists \writeEE{x}{k} \not= \writeEE{x}{i}. (\notHB{\readEE{x}{j}}{\writeEE{x}{k}} \wedge \hb{\writeEE{x}{i}}{\writeEE{x}{k}} ) \}
 \\
 W_2 & = & \{ \writeEE{x}{i} \mid \writeEE{x}{i} \in T \wedge \hb{\writeEE{x}{i}}{\readEE{x}{j}} \wedge
 \\ && \ \ \ \ \ \ \ \ \ \ 
  \neg \exists \writeEE{x}{k} \not= \writeEE{x}{i}. (\hb{\writeEE{x}{k}}{\readEE{x}{j}} \wedge \hb{\writeEE{x}{i}}{\writeEE{x}{k}}) \}
 \eda
  Let $W = W_1 \cup W_2$.
  Then, we refer to $W$ as the \emph{WRD candidates}
  w.r.t.~$\readEE{x}{j}$. 
\end{definition}
The set $W_1$ covers all writes that are unsynchronized with respect to the read.
We refer to $W_1$ as the \emph{unsynchronized} WRD candidates.
The set $W_2$ covers all writes that happen before that read.
We refer to $W_2$ as the \emph{synchronized} WRD candidates.
For both cases, the write must be \emph{closest} to the read.
In case of $W_1$, this is achieved by side condition
$\neg \exists \writeEE{x}{k} \not= \writeEE{x}{i}. (\notHB{\readEE{x}{j}}{\writeEE{x}{k}} \wedge \hb{\writeEE{x}{i}}{\writeEE{x}{k}} )$
and in case of $W_2$ by the side condition
$\neg \exists \writeEE{x}{k} \not= \writeEE{x}{i}. (\hb{\writeEE{x}{k}}{\readEE{x}{j}} \wedge \hb{\writeEE{x}{i}}{\writeEE{x}{k}})$.
Depending on the context, we write $W(e), W_1(e), W_2(e)$ to highlight
that we consider the WRD candidates for a specific read event~$e$.

\begin{example}
  Consider the following trace.
  \bda{ll|l|l}
  & \thread{1}{} & \thread{2}{} & \thread{3}{}
  \\ \hline
  1. & \writeE{x} &&
  \\
  2. & && \writeE{x}
  \\
  3. && \lockE{y}
  \\
  4. && \writeE{x}
  \\
  5. && \unlockE{y}
  \\
  6. & \lockE{y} &
  \\
  7. & \unlockE{y} &
  \\
  8. & \readE{x} &
  \eda
  Consider the read event $\readEE{x}{8}$.
  The WRD candidates as characterized by Definition~\ref{def:wrd-candidates}
  are as follows:
  $W_1 = \{ \writeEE{x}{2} \}$ and $W_2 = \{ \writeEE{x}{1}, \writeEE{x}{4}\}$.
  
  Event $\writeEE{x}{2} \in W_1$ because $\hb{\readEE{x}{8}}{\writeEE{x}{2}}$
  and there is no other write that happens in between.
  Event $\writeEE{x}{2}$ is the closest write based on the trace position.
  Hence, this is the (sole) WRD candidate as characterized by Definition~\ref{def:schedulable-happens-before}.

  In addition, we consider $\writeEE{x}{1}$ and  $\writeEE{x}{4} \in W_2$.
  Both writes are unsynchronized with respect to $\readEE{x}{8}$ as well as
  $\writeEE{x}{2}$.
\end{example}

\begin{proposition}
  \label{prop:prop}  
The write event as characterized by Definition~\ref{def:schedulable-happens-before}
must be in one of the sets $W_1$ and $W_2$.
\end{proposition}

\begin{proposition}
  \label{prop:prop-one}
  Let $T$ be a trace and $e$ be some read event.
  The number of elements in $W_1(e)$ is $O(k)$ where $k$ is the number of threads.
  The same applies to $W_2(e)$.
\end{proposition}

The set $W_1 \cup W_2$ is a conservative approximation of all potential writes that may affect the read.
Instead of imposing the WRD condition for a single write candidate
as done for SHB, we consider the following variants.
We (a) impose all WRD candidates, (b) pick one of the WRD candidates as characterized by Definition~\ref{def:wrd-candidates}.
First, we consider variant~(a).

\begin{definition}[Strong Schedulable Happens-Before]
\label{def:strong-schedulable-happens-before}  
  Let $T$ be a trace.
  We define a relation $\strongshb{}{}$ among trace events
  as the smallest partial order
  that satisfies
  PO and RAD from Definition~\ref{def:happens-before}
    and the following additional condition:
  \begin{description}
  \item[Strong write-read dependency (StrongWRD):]
    Let $\readEE{x}{j}\in T$ and $W(\readEE{x}{j})$ its set of WRD candidates.
    Then, $\strongshb{e}{\readEE{x}{j}}$ for each $e \in W(\readEE{x}{j})$.
  \end{description}
  We refer to $\strongshb{}{}$ as the \emph{strong schedulable happens-before}
  relation derived from trace $T$.
\end{definition}

A strong schedulable happens-before relation may not exist.

\begin{example}
\label{ex:no-strong-shb}  
  Consider the trace
  $T = [
    \thread{1}{\readEE{y}{1}},
    \thread{1}{\writeEE{x}{2}},
    \thread{2}{\readEE{x}{3}},
    \thread{2}{\writeEE{y}{4}} ]$.
  Based on the program order, we find
  $\hb{\readEE{y}{1}}{\writeEE{x}{2}}$ and
  $\hb{\readEE{x}{3}}{\writeEE{y}{4}}$.
  Furthermore, $W(\readEE{y}{1}) = \{ \writeEE{y}{4} \}$
  and $W(\readEE{x}{3}) = \{ \writeEE{x}{2} \}$.
  The strong schedule happens-before relation additionally imposes
  $\strongshb{\writeEE{y}{4}}{\readEE{y}{1}}$ and
  $\strongshb{\writeEE{x}{2}}{\readEE{x}{3}}$.
  In combination with the HB relations we encounter a cycle.
  Consider
  $ \readEE{y}{1} \hbSym\
  \writeEE{x}{2} \strongshbSym\
  \readEE{x}{3} \hbSym\
  \writeEE{y}{4} \strongshbSym\
  \readEE{y}{1}$.
  This contradicts the properties of a partial order.
  Hence, the strong schedulable happens-before relation does not exist
  for this case.
\end{example}

Next, we consider variant~(b).
Instead of imposing all WRD candidates,
we (arbitrarily) pick one of the candidates.

\begin{definition}[Some Schedulable Happens-Before]
\label{def:Some-schedulable-happens-before}  
  Let $T$ be a trace.
  We define a relation $\someshb{}{}$ among trace events
  as the smallest partial order
  that satisfies
  PO and RAD from Definition~\ref{def:happens-before}
    and the following additional condition:
  \begin{description}
  \item[Some write-read dependency (SomeWRD):]
    Let $\readEE{x}{j}\in T$ and $W(\readEE{x}{j})$ its set of WRD candidates.
    Then, $\someshb{e}{\readEE{x}{j}}$ for some $e \in W(\readEE{x}{j})$.
  \end{description}
  We refer to $\someshb{}{}$ as the \emph{some schedulable happens-before}
  relation derived from trace $T$.
\end{definition}

We expect that an instance of $\someshb{}{}$ always exists.
This is only the case if there is an initial write for each read.
Recall Example~\ref{ex:no-strong-shb}.
For each read, there is only one choice of WRD candidate.
As the example shows, the resulting order relation has a cycle.

To ensure that an instance of $\someshb{}{}$ exists,
we demand that for each read there must be always a write happening before that read.
Then, $\shb{}{}$ is a specific instance of a $\someshb{}{}$.

\begin{definition}[Initial Write]
  Let $T$ be a trace and $\hb{}{}$ the happens-before relation derived from $T$.
  We say that each read event $e$ in $T$ enjoys an \emph{initial write}
  if $W_2(e) \not = \{ \}$.
\end{definition}

\begin{proposition}
  \label{prop:prop-two}
  Let $T$ be a trace where each read event enjoys an initial write.
  Then, $\shb{}{}$ is an instance of $\someshb{}{}$.
\end{proposition}

Besides $\shb{}{}$, we also show that $\hb{}{}$ is an instance.

\begin{proposition}
\label{prop:someshb-hb-equivalence}  
  Let $T$ be a trace where each read event enjoys an initial write.
  Then, $\hb{}{}$ is an instance of $\someshb{}{}$.
\end{proposition}

In general, the resulting $\someshb{}{}$ relations are distinct
from $\hb{}{}$ and $\shb{}{}$.

\begin{example}
  \label{ex:not-all-wrd-combos-yield-someshb}
  Consider the following trace (on the left)
  and the resulting WRD candidates
  and $\someshb{}{}$ instances (on the right).

\bda{lccl}  
  \ba{ll|l}
  & \thread{1}{} & \thread{2}{}
  \\ \hline
  1. & & \writeE{x}
  \\
  2. & \writeE{y} &
  \\
  3. & \readE{y} &
  \\
  4. & \writeE{x} &
  \\
  5. & & \readE{x}
  \\
  6. & & \writeE{y}
  \ea
  &\ \ \ \ \ &&
  \ba{l}
    W(\readEE{y}{3}) = \{ \writeEE{y}{2}, \writeEE{y}{6} \}
    \\ W(\readEE{x}{5}) = \{ \writeEE{x}{1}, \writeEE{x}{4} \}
    \\
    \\ 1. \ \ \someshb{\writeEE{y}{2}}{\readEE{y}{3}} \ \
              \someshb{\writeEE{x}{1}}{\readEE{x}{5}}
    \\ 2. \ \ \someshb{\writeEE{y}{2}}{\readEE{y}{3}} \ \
              \someshb{\writeEE{x}{4}}{\readEE{x}{5}}
    \\ 3. \ \ \someshb{\writeEE{y}{6}}{\readEE{y}{3}} \ \
              \someshb{\writeEE{x}{1}}{\readEE{x}{5}}
  \ea
\eda  
  The first instance corresponds to $\hb{}{}$ and
  the second instance to $\shb{}{}$.
  The third instance is different from HB and SHB.
\end{example}  

The example shows that not all combinations (four in our case) are feasible.
For $\someshb{\writeEE{y}{6}}{\readEE{y}{3}}$ 
the only choice left is $\someshb{\writeEE{x}{1}}{\readEE{x}{5}}$
because $\someshb{\writeEE{x}{4}}{\readEE{x}{5}}$ leads to a cycle.

%%%%%%%%%%%%%%%%%%%%%%%%%%%%%%%%%%%%%%%%%%%%%%%%%%%%%%%%%%%%%%
%%%%%%%%%%%%%%%%%%%%%%%%%%%%%%%%%%%%%%%%%%%%%%%%%%%%%%%%%%%%%%
\section{Data Race Diagnosis}
\label{sec:diagnosis}

We make use of the results of the previous section
to provide diagnostic information to examine if data races reported
are affected by inaccurately traced write-read dependencies.
%% MS: omit
%% For a write-write data race events
%% involved are always unsynchronized with respect to each other.
%% In case of write-read races, we additionally need to consider
%% write-read dependencies.
For brevity, we only consider write-write races.
Write-read cases 
will be covered when discussing our implementation in the upcoming section.

\begin{definition}[Write-Write Data Races]
  Let $T$ be a trace where
  $e, f \in T$ are two write events on the same variable.
  We say $(e,f)$ are in a \emph{write-write HB data race}
  if neither $\hb{e}{f}$ nor $\hb{f}{e}$.

  We say $(e,f)$ are in a \emph{write-write \strongSHBRel\ data race}
  if neither $\strongshb{e}{f}$ nor $\strongshb{f}{e}$.

  Let $\someshb{}{}$ be a specific choice of a some schedulable happens-before relation.
  We say $(e,f)$ are in a \emph{write-write \someSHBRel\ data race}
  if neither $\someshb{e}{f}$ nor $\someshb{f}{e}$.
\end{definition}

We wish to answer the following questions.

\begin{enumerate}
    \item {\bf Q1:} Given a HB data race, is this a race for any instance of $\someshb{}{}$?
  \item {\bf Q2:} Given a HB data race, does an instance of $\someshb{}{}$ exists under which
  the race goes away?
\end{enumerate}

A yes to the first question means inaccurate tracing had no impact.
A yes to the second question means that we cannot be sure.
To check both questions we construct a graph from the trace.

\begin{definition}
\label{def:graph-edges}
Let $T$ be a trace.
We assume that for each read event $e \in T$
we have available the set $W(e)$ of WRD candidates.
We define a graph $G=(N,E)$ with a set $N$ of nodes
and set $E$ of directed edges as follows:
\begin{itemize}
\item $N = \{ e \mid e \in T \}$.
\item $E = \{ e \starrow{HB} f \mid e, f \in T \wedge \hb{e}{f} \}
    \cup \{ e \starrow{W(r)} f \mid \mbox{$r$ read event} \wedge e \in W(r) \}$.
\end{itemize}
We refer to $G$ as the \emph{graph derived from $T$}.

Let $e, f \in N$. We say that there is a \emph{path} from $e$ to $f$,
if we find $g_1,\dots,g_n \in N$
such that $e \rightarrow g_1 \rightarrow \dots \rightarrow g_n \rightarrow f$
where all nodes are distinct.
\end{definition}  

Edges are derived from either a HB relation or some WRD candidate.
They are labeled to keep track of their origin.
If this information does not matter, we write $e \rightarrow f$ for short.

The graph includes all possible combinations of WRD candidates.
If there is no path between two events involved in a data race,
we can answer {\bf Q1} positively.

\begin{proposition}[Guaranteed Data Race]
\label{prop:guaranteed-data-race}  
  Let $T$ be a trace and $G$ be the graph derived from $G$.
  Let $(e,f)$ be a pair of events that are in a write-write HB data race.
  If there is no path between $e$ and $f$ then
  for any choice of $\someshb{}{}$ we have that
  $(e,f)$ are in a write-write \someSHBRel\ data race.
\end{proposition}

The above statements only applies to those $\someshb{}{}$ relations that actually
constitute a partial order. Recall that no every choice of WRD candidates
results in a partial order and $\strongshb{}{}$ does not necessarily exist.
See Examples~\ref{ex:no-strong-shb} and~\ref{ex:not-all-wrd-combos-yield-someshb}.

To answer {\bf Q2} we need to find path for a particular choice of WRD candidates.

\begin{proposition}[Maybe Data Race]
\label{prop:maybe-data-race}  
    Let $T$ be a trace and $G$ be the graph derived from $G$.
    Let $(e,f)$ be a pair of events that are in a write-write HB data race.
    We assume that there exists a path among $e$ and $f$
    with distinct nodes and the label $W(r)$ appears at most once.
    Then, there exists $\someshb{}{}$ under which $e$ and $f$ are not in a race.
\end{proposition}

%% MS: omit, don't need this technical result
%% \begin{proposition}
%%   Let $T$ be a trace and $e$ be some read event.
%%   Let $f, g \in W_1(e)$ where $\hb{f}{g}$.
%%   Let $\someshbN{1}{}{}$ be built by choosing $f$
%%   and $\someshbN{2}{}{}$ be built by choosing $g$.
%%   Otherwise, the choice of WRD candidates for
%%   $\someshbN{1}{}{}$ and $\someshbN{2}{}{}$ is the same.
%%   Then, $\someshbN{2}{}{}$ subsumes $\someshbN{1}{}{}$
%%   when checking if two events are in a happens-before relation.
%% \end{proposition}  
%% \begin{proof}
%%   We have that $\someshbN{1}{f}{e}$ and need to show
%%   that $\someshbN{2}{f}{e}$ holds as well.
%%   We have that $\someshbN{2}{g}{e}$.
%%   Based on $\hb{f}{g}$ we can also conclude $\someshbN{2}{f}{e}$.
%%   This is the only case to consider, otherwise
%%   $\someshbN{1}{}{}$ and $\someshbN{2}{}{}$ agree.
%%   \qed
%% \end{proof}  

%%%%%%%%%%%%%%%%%%%%%%%%%%%%%%%%%%%%%%%%%%%%%%%%%%%%%%%%%%%%%%
%%%%%%%%%%%%%%%%%%%%%%%%%%%%%%%%%%%%%%%%%%%%%%%%%%%%%%%%%%%%%%
\section{Implementation and Experiments}
\label{sec:experiments}

\begin{algorithm}
    \caption{Extended FastTrack algorithm (\SSHB\ algorithm)}\label{alg:vc-race-pairs}

\begin{tabular}{ll}
\begin{minipage}{.6\textwidth}      
            \begin{algorithmic}[1]
                \Procedure{acquire}{$i,x$}
                \State $\threadVC{i} = \supVC{\threadVC{i}}{\lockVC{x}}$
                \State $edges = edges \cup \{\lockVC{x} \gtEdge \thread{i}{\accVC{\threadVC{i}}{i}}\}$
                \EndProcedure
            \end{algorithmic}
\end{minipage}

&

\begin{minipage}{.4\textwidth}
            \begin{algorithmic}[1]
                \Procedure{release}{$i,x$}
                \State $\lockVC{x} = \threadVC{i}$
                \State $\incC{\threadVC{i}}{i}$
                \EndProcedure
            \end{algorithmic}
\end{minipage}
\end{tabular}

            \begin{algorithmic}[1]
                \Procedure{write}{$i,x$}
                \For{$\thread{j}{k} \in cw(x)$}
                \If {$k > \accVC{\threadVC{i}}{j}$}
                \State $races = races \cup \{(\thread{j}{k}, \thread{i}{\accVC{\threadVC{i}}{i}})\}$
                \EndIf
                \EndFor
                \State $cw(x) = \{ \thread{i}{\accVC{\threadVC{i}}{i}} \}
                \cup \{ \thread{j}{k} \mid \thread{j}{k} \in cw(x) \wedge
                k > \accVC{\threadVC{i}}{j} \}$
                
                \For{$\thread{j}{k} \in cr(x)$}
                \If {$k > \accVC{\threadVC{i}}{j}$}
                \State $edges = edges \cup \{\thread{i}{\accVC{\threadVC{i}}{i}} \gtEdge \thread{j}{k} \}$
                \State $races = races \cup \{(\thread{j}{k}, \thread{i}{\accVC{\threadVC{i}}{i}})\}$
                \EndIf
                \EndFor
                \State $\incC{\threadVC{i}}{i}$
                \EndProcedure
            \end{algorithmic}                       
            \begin{algorithmic}[1]
                \Procedure{read}{$i,x$}
                \For{$\thread{j}{k} \in cw(x)$}
                \If {$k > \accVC{\threadVC{i}}{j}$}
                \State $edges = edges \cup \{\thread{j}{k} \gtEdge \thread{i}{\accVC{\threadVC{i}}{i}}\}$
                \State $races = races \cup \{(\thread{i}{\accVC{\threadVC{i}}{i}}, \thread{j}{k})\}$
                \EndIf
                \EndFor
                \State $cr(x) = \{ \thread{i}{\accVC{\threadVC{i}}{i}} \}
                \cup \{ \thread{j}{k} \mid \thread{j}{k} \in cr(x) \wedge
                k > \accVC{\threadVC{i}}{j} \}$               
                \State $\incC{\threadVC{i}}{i}$
                \EndProcedure
            \end{algorithmic}
    
\end{algorithm}

%% MS: omit, has been largely addressed, kept for reference
%% 
%% \ms{more} assumes offline approach, our diagnosis comes with some overhead.
%% while processing the trace and then when examining the race results reported
%% by FastTrack and SHB.
%% 
%% \ms{todo}
%% Inaccurate trace means no guarantee that WRDs are correct.
%% This affects accuracy of SHB. False positives and false negatives are possible.
%% 
%% FastTrack does not impose WRD. so some inaccurate trace
%% should not have any impact on FastTrack.
%% Cause no WRDs imposed, false positives may arise.
%% FastTrack only sound up to the first race reported.
%% For further races reported, how can use our analysis to
%% check if this race is a potential false positive?
%% So guaranteed race regardless of WRD and the candidate used.
%% Can apply our diagnostic methods to identify such cases.
%% 
%% We only examine reported races.
%% Question we don't consider (yet). We have possibly missed a data race.
%% 
%% Include lockset.
%% 
%% Our methods allows for a detailed comparison of the
%% results reported by FastTrack and SHB.

%%%%%%%%%%%%%%%%%%%%%%%%%%%%%%%%%%%%%%%%%%%%%%%%%%%%%%%%%%%%%%
\subsection{Implementation}

To implement the diagnostic methods described in the previous section,
we make use of an extended version of the FastTrack algorithm from \cite{flanagan2010fasttrack}
to (1) identify HB data race pairs and to (2) build up the graph
of HB and WRD candidate edges (see Definition~\ref{def:graph-edges}).
See Algorithm~\ref{alg:vc-race-pairs}.

We maintain several vector clocks.

\begin{definition}[Vector Clocks]
  A \emph{vector clock} $V$ is a list of \emph{time stamps} of the following form.
  \bda{rcl}
   V  & ::= & [i_1,\dots,i_n]
   \eda
   We assume vector clocks are of a fixed size $n$.
   Time stamps are natural numbers and each time stamp position $j$ corresponds to the thread
   with identifier $j$.

 We write $\accVC{V}{j}$ to access the time stamp at position $j$.
 We write $\incC{V}{j}$ as a short-hand for incrementing the vector clock $V$ at position $j$ by one.
\end{definition}

For each thread $i$ we maintain a vector clock $\threadVC{i}$.
For each mutex $y$, we find vector clock $\lockVC{y}$ to maintain the last release event on $y$.
Initially, for each vector clock $\threadVC{i}$
all time stamps are set to 0 but position $i$ where the time stamp is set to 1.
For $\lockVC{y}$ all time stamps are set to~0.

To efficiently compare read and write events, we make use of epochs.
\begin{definition}[Epoch]
  Let $j$ be a thread id and $k$ be a time stamp.
  Then, we write $\thread{j}{k}$ to denote an \emph{epoch}.
\end{definition}
Each event can be uniquely associated to an epoch.
Take its vector clock and extract the time stamp $k$ for the thread $j$
the event belongs to. For each event this pair of information
represents a unique key to locate the event.

Epochs are use to manage unsynchronized reads and writes.
Like in the original FastTrack algorithm, we find
a set $cr$ of unsynchronized reads where
each read is represented by its epoch $\thread{j}{k}$.
In addition, we also maintain a set $cw$ of unsynchronized writes.
This set serves the purpose to identify all HB write-write and write-read race pairs
and to build up the edges for WRD candidates.
Race pairs are recorded in $races$ and edges in the set $edges$.

We consider the various cases.
In case of an acquire event we synchronize the thread's vector clock with the most recent release event
by building the union of the vector clocks $\threadVC{i}$ and $\lockVC{x}$. In case of a release event, we update $\threadVC{x}$.

Each write event compares the epochs in $cw$ against the thread's vector clock.
The condition $k > \accVC{\threadVC{i}}{j}$ (line 3) indicates that the write is concurrent to the current epoch (event) in $cw$
and therefore a new write-write race pair is added to $races$. All data race pairs are represented via their epochs.
We update the set $cw$ (line 7) and only keep writes in $cw$ that are concurrent.
For each concurrent read in $cr$ we add a write-read race pair (line 11).
For each WRD candidate of read, we add a new edge (line 10). This step corresponds to the construction in Definition~\ref{def:graph-edges}).

By updating $cw$ (line 7) we include an optimization.
We only maintain WRD candidates that are unsynchronized with respect to each other.
We know candidates in $W_1$ are unsynchronized with respect to each other.
The same applies to $W_2$.
However, it is possible that
$f \in W_1$ and $g \in W_2$ where $\hb{g}{f}$.
Due to space limitations, we refer to Appendix~\ref{sec:auxiliary} for an example.
Event $g$ can be safely omitted, because any WRD relation introduced by $g$
is covered by $f$.

The treatment is similar for read events.
We assume that for events appearing in the same thread,
we add an edge to the set $edges$ if they appear one after the other (in that thread).
For brevity, we omit the details.

\begin{example}
    We consider a run of the \SSHB{} algorithm by processing the following trace. We underline events for which a new race pair is detected. The subscript is the vector clock for each event. The set $\{\thread{1}{1}\}_x$ in column $cw$, depicts the set $cw$ for variable $x$. Similarly for column $cr$, $\{\thread{2}{1}\}_x$ depicts the set $cr$ for variable $x$.
    
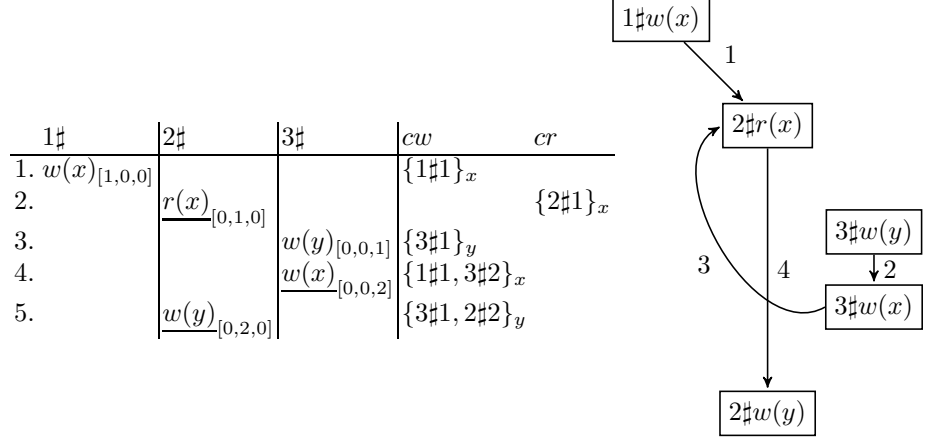
    \begin{tabular}{ll}
      \hspace{-.5cm}
    \begin{minipage}{0.7\linewidth}
        \bda{ll|l|l|ll}
        & \thread{1}{} & \thread{2}{} & \thread{3}{} & cw & cr\\ \hline
        1. & \writeE{x}_{[1,0,0]} & \textbf{} & \textbf{} & \{\thread{1}{1}\}_x\\
        2. & &              \underline{\readE{x}}_{[0,1,0]}  & \textbf{} & \textbf{} & \{\thread{2}{1}\}_x\\
        3. & & &                           \writeE{y}_{[0,0,1]} & \{\thread{3}{1}\}_y\\
        4. & & &                           \underline{\writeE{x}}_{[0,0,2]} & \{\thread{1}{1}, \thread{3}{2}\}_x\\
        5. & &              \underline{\writeE{y}}_{[0,2,0]} & \textbf{} & \{\thread{3}{1}, \thread{2}{2}\}_y
        \eda
    \end{minipage}
    &
          \hspace{-.5cm}
    \begin{minipage}{0.3\linewidth}
        \begin{tikzpicture}[->,>=stealth',shorten >=1pt,auto,node distance=2cm,semithick]
            \node[draw=black] (A) {$\thread{1}{\writeE{x}}$};
            \node[draw=black] (B) [below right of=A] {$\thread{2}{\readE{x}}$};
            \node[draw=black] (C) [below right of=B] {$\thread{3}{\writeE{y}}$};
            \node[draw=black] (D) [below of=C, yshift=1cm] {$\thread{3}{\writeE{x}}$};
            \node[draw=black] (E) [below left of=D] {$\thread{2}{\writeE{y}}$};
            
            \path (A) edge node {1} (B);
            \path (D.west) edge [bend left=90] node {3} (B.west);
            \path (B) edge node {4} (E);
            \path (C) edge node {2} (D);
        \end{tikzpicture}
    \end{minipage}

\end{tabular}

The read in the second step is concurrent to the write in the first step. This is detected by comparing the epochs in $cw_x$ with the vector clock of thread two. Since the read is concurrent to the write, a new edge is added in the graph from the write to the read event (edge 1).

The same applies to the write in the fourth step. It is concurrent to the read in step 2, thus a new data race pair is created and a edge is added in the graph (edge 3).

At step five we create the next race pair between the events $\thread{2}{\writeE{y}}$ and $\thread{3}{\writeE{y}}$. In the first phase all write-read dependencies are ignored. Thus the two writes on $y$ are concurrent. This race pair is a false positive since a write-read dependency exists under which the events are ordered.

In the post processing phase the detected data race pairs $(\thread{1}{\writeE{x}},\thread{2}{\readE{x}})$, $(\thread{2}{\readE{x}}, \thread{3}{\writeE{x}})$ and $(\thread{3}{\writeE{y}}, \thread{2}{\writeE{y}})$ are tested. 

For the first two data race pairs we omit the trivial path that exists due to the write-read dependencies in the implementation. We do not find an alternative path between the events without the direct write-read dependency edge and thus report a guaranteed data race for both data race pairs.

For the data race pair $(\thread{3}{\writeE{y}}, \thread{2}{\writeE{y}})$ we find a path from event $\thread{3}{\writeE{y}}$ to $\thread{2}{\writeE{y}}$ by following edge two, three and four. Thus the race pair is categorized as maybe.  
    
\end{example}

We consider the space and time complexity of running Algorithm~\ref{alg:vc-race-pairs}.
We assume $n$ to be the size of the trace and $k$ to be the number of threads.
The size of $cr$ and $cw$ is bound by $O(k)$ and
the size of $races$ and $edges$ is bound by $O(n*n)$.
We assume the addition and removal of elements to any of the sets takes time $O(1)$.
For write, there are $O(k)$ elements to consider. The same applies to read.
The update of a vector clock takes time $O(k)$.
So, processing of each event takes time $O(k)$.
Hence, the running time is $O(n*k)$.

The diagnostic analysis phase then checks for each write-write data race pair $(\thread{i}{k}, \thread{j}{k'}) \in races$
if there is path among $\thread{i}{k}$ and $\thread{j}{k'}$
in the graph derived from $edges$. 
No path guarantees the race (see Proposition~\ref{prop:guaranteed-data-race} )
whereas some path means `maybe' (see Proposition~\ref{prop:maybe-data-race}  ).
The time complexity is $O(n*n + n) = O(n*n)$ to check if a path exists.
There are $O(n*n)$ race pairs. So, overall the time to check if races are guaranteed or maybe races
is $O(n^4)$. In practice, we experienced no performance issues as the worst case never arises.

The implementation can detect write-read and read-write data race pairs.
A read and write event are in a guaranteed data race
if there is no path between the events when omitting the write-read dependency edge that connects them. 

\begin{minipage}{\linewidth}
    \small 
    
    \begin{tabular}{cccccccc}
        \hline
        \multicolumn{1}{|l|}{\textbf{Program}} & \multicolumn{1}{l|}{\textbf{Events}} & \multicolumn{1}{l|}{\textbf{Threads}} & \multicolumn{1}{l|}{\textbf{Vars}} & \multicolumn{1}{l|}{\textbf{Locks}} & \multicolumn{1}{l|}{\textbf{Reads}} & \multicolumn{1}{l|}{\textbf{Writes}} & \multicolumn{1}{l|}{\textbf{Syncs}} \\ \hline
        moldyn & 53308289 & 4 & 18423 & 1 & 45153928 & 8154330 & 11\\
        raytracer & 224598 & 4 & 5 & 5 & 224599 & 1 & 26\\
        xalan & 62886984 & 17 & 9707 & 974 & 5268214 & 1210737 & 446450\\
        lusearch & 2659371 & 17 & 124 & 772 & 1132 & 271  & 1328968\\
        tomcat & 26450441 & 58 & 42949 & 20136 & 10432626 & 383212 & 404693 \\
        avrora & 16671631 & 7 & 479049 & 7 & 11847735 & 1868321 & 1477776\\
        h2 & 360617324 & 18 & 749954 & 48 & 95939995 & 698490 & 1680748\\
%% MS: omit, ignore filter, shared/unshared variables, cap of events
%%         moldyn & 53308289* & 4 & 18423 & 1 & 45153928 & 8154330 & 11\\
%%         raytracer & 224598* & 4 & 5 & 5 & 224599 & 1 & 26\\
%%         xalan & 62886984 & 17 & 9707 & 974 & 5268214 & 1210737 & 446450\\
%%         lusearch & 2659371* & 17 & 124 & 772 & 1132 & 271  & 1328968\\
%%         tomcat & 26450441 & 58 & 42949 & 20136 & 10432626 & 383212 & 404693 \\
%%         avrora & 16671631* & 7 & 479049 & 7 & 11847735 & 1868321 & 1477776\\
%%         h2 & 360617324* & 18 & 749954 & 48 & 95939995 & 698490 & 1680748\\        
    \end{tabular}
    \captionof{table}{Benchmarks and some meta data}
    \label{tab:bench:disc}
\end{minipage}

%%%%%%%%%%%%%%%%%%%%%%%%%%%%%%%%%%%%%%%%%%%%%%%%%%%%%%%%%%%%%%
\subsection{Experimental Results}

For benchmarking we use two Intel Xeon E5-2650 and 64 gb of RAM with Ubuntu 18.04 as operating system.
We use tests from the Java Grande Forum \cite{smith2001parallel} and from the DaCapo (version 9.12, \cite{Blackburn:2006:DBJ:1167473.1167488}) benchmark suite. All benchmark programs are written in Java and use up to 58 threads. Details for each benchmark can be found in Table \ref{tab:bench:disc}.
Column Events contains the length of the trace. Columns Threads, Vars and Lock describe the total number of threads, variables and mutex
that were used in the observed schedule. Columns Reads, Writes and Syncs describe the kind of events that are stored in the trace (except fork/join events).%% For benchmarks marked with $*$ the described filter was used to ignore unshared variables.
%%For the h2 benchmark it was necessary to cap the events at 30M events due to memory restrictions.

We compare our algorithm SSHB and its diagnostic (post-processing) phase
against variants of FastTrack and SHB.
FastTrack and SHB only report code locations that are in a race.
Hence, we consider some variants of FastTrack, referred to as \HBPartner{},
and SHB, referred to as \SHBPartner{}, that report pairs of events that
are in a race. We ignore pairs that refer to the same source code locations.

\begin{minipage}{\linewidth}
    \small
    \begin{tabular}{c|c|c|c|c|c|c|c}
        \hline
        \textbf  & moldyn  & raytracer & xalan & lusearch & tomcat & avrora & h2  \\ \hline
        
        \textbf{\HBPartner} (s) & 97 & 0 & 94 & 3 & 63 & 38 & 81 \\
        \#Races  & 45 &  1 & 45 & 23 & 1322 & 36 & 480\\
        \#Read-write & 19 & 0 & 11 & 1 & 303 & 6 & 95\\
        \#Write-read & 16 & 1 & 22 & 22 & 685 & 23 & 205\\
        \#Write-write& 10 & 0 & 12 & 0 & 334 & 7 & 180\\
        \hline
        
        \textbf{\SHBPartner} (s) & 104 & 0 & 95 & 3 & 67 & 38 & 83 \\
        \#Races & 18 &  1 & 43 & 22 & 662 & 24 & 208 \\ 
        \#Read-write & 6 & 0 & 11 & 1 & 236 & 6 & 73  \\
        \#Write-read & 12 & 1 & 20 & 21 & 221 & 18 & 123\\
        \#Write-write & 0 & 0 & 12 & 0 & 205 & 0 & 12\\
        \hline
        
        \textbf{\SSHB} (s) & 249 & 0 & 100 & 9 & 423 & 52 & 224\\
        Phase1+2 (s) & 69+70 & 0+0 & 12+3 & 3+2 & 32+355 & 28+7 & 46+124\\        
        \CountRaces/G & 45/6 &  1/1  & 45/34  & 23/22  & 1322/194  & 36/20 & 480/120\\
        \#Read-write/G & 19/6 & 0/0 & 11/10 & 1/1 & 303/91 & 6/4 & 95/37\\
        \#Write-read/G & 16/0 & 1/1 & 22/13 & 22/21 & 685/51 & 23/16 & 205/78\\
        \#Write-write/G & 10/0 & 0/0 & 12/11 & 0/0 & 334/52 & 7/0 & 180/5\\
        \#w(r)(Avg) & 1.72 & 1.00 &  1.00 & 1.21 & 1.05 & 1.65 & 1.06\\
        \#w(r)(Max) & 6683 & 1 & 12 & 2 & 190 & 6 & 19\\
        \hline
    \end{tabular}
    \captionof{table}{
          Benchmark results}
    \label{tab:bench:races}
\end{minipage}

For \HBPartner{} and \SHBPartner{} we report the overall processing time in seconds (s)
which includes parsing the trace, running the algorithms etc.
The number 0 means that the time measured is below 1 second.
We make use of
an extension of RoadRunner for creating a log file of all trace events.
Our algorithms are implemented in Go, therefore, we need to apply some transformations
that creates some extra overhead. In the first row, labeled with the algorithm name, we present the time in seconds of the complete run including parsing the trace, race prediction and reporting. The row labeled \#Races shows the amount of predicted data race pairs. We only report and count unique code location pairs.
Race pairs are additionally categorized in read-write,
write-read and write-write data races. 
The row \#Read-write represents the case that the read appeared before the write in the trace,
whereas \#Write-read represents the case that the write appeared first.

For SSHB we separately measure the time to run Algorithm~\ref{alg:vc-race-pairs} (first phase)
and to check for guaranteed/maybe data races in the (second) diagnostic phase.
See row labeled Phase1 + 2 (s).
We present the number of races reported in the first phase and the number of guaranteed races
obtained via the (second) diagnostic phase.
For example, for the moldyn benchmark we find in row \CountRaces/G
the entry 45/6 which indicates 45 race pairs in the first phase of which are 6 guaranteed races.
The leading number, here 45, always coincides with the number reported by \HBPartner{}.
We also report the average and maximum size of the set of WRD candidates
via \#w(r)(Avg) and \#w(r)(Max).
Table \ref{tab:bench:races} shows the result for each benchmark.

Guaranteed data races are likely to be the more `critical' races and are not influenced by inaccurate tracing.
Consider the tomcat benchmark.
\HBPartner{} reports 1322 data races whereas \SHBPartner{} still reports 662.
Recall that \SHBPartner{} imposes a fixed write-read dependency based on the order of writes/reads in the trace.
\SSHB{} reports 194 guaranteed races by considering all possible WRD candidates due to inaccurate tracing.
This is still a high number but it becomes more feasible for the user
to trace down 194 races than 1322 or 662.

Based on the average number of WRD candidates, see row labeled $\#w(r)(Avg)$,
we argue that imposing a fixed write-read dependency as done by \SHBPartner{} may lead to inaccuracies.
For example, for benchmark moldyn we find 1.72 WRD candidates and for benchmark avrora we find 1.65 WRD candidates on the average.
This is another indication to focus on guaranteed races rather than any race.

For the benchmark moldyn we find a maximum of 6683 WRD candidates, see row labeled \#w(r)(Max).
The reason for this high number is that the associated read event has multiple data races with write events that result from the same code location.
Recall that we only count the number of race pairs with distinct code locations.
However, our graph captures all events regardless if they refer to the same code location.
Hence, the significant difference between \#w(r)(Max) and the number of races reported.

The traces for our benchmark programs are rather long.
So, we did not inspect the traces to check for inaccuracies regarding unsynchronized reads and writes nor did we check if maybe races are false positives.
For example, consider the avrora benchmark.
\HBPartner{} reports 36 races whereas \SHBPartner{} reports 24 races.
Our method \SSHB{} reports that out of the 36 (24) there are 20 guaranteed races.
To check that the 16 \HBPartner{} races and 4 \SHBPartner{} races are false positives,
would require manual inspection of the specific program run as well as detailed knowledge
of the program text.

To summarize.
We observe that across all benchmarks there is a difference
between the number of races and the number of guaranteed races.
Furthermore, the max number of WRD candidates, see \#w(r)(Max), is besides the raytracer and lusearch
benchmark fairly high. Hence, it is likely that some traces are inaccurate.
Our method identifies guaranteed races and thus
allows the user to focus on `critical' races.

%%%%%%%%%%%%%%%%%%%%%%%%%%%%%%%%%%%%%%%%%%%%%%%%%%%%%%%%%%%%%%
\subsection{Impact of RoadRunner on Results}
\label{sec:roadrunner-results}

In our experiments, we use the RoadRunner version from \url{https://github.com/stephenfreund/RoadRunner}.
For this version, acquire and release events are issued \emph{after} the respective operation has executed. Then, guarantee {\bf INST\_RA}, see Section~\ref{sec:run-time-events}, no longer holds
and we may encounter false positives among
the guaranteed races reported by \SSHB{}.

Recall the example in Figure~\ref{fig:example-lock-false-positive}.
The release event is only issued after the actual operation.
There is some subsequent acquire operation that takes place in some concurrent thread.
As the resulting acquire event is issued concurrently to the release event,
the acquire event might appear before release in the trace.
In our example, $\lockE{y}$ appears at trace position~$3$ and $\unlockE{y}$ appears at trace position~$5$.
When deriving the happens-before relation based on the inaccurate trace,
the two write events at trace position~$2$ and~$4$ are unordered.
Hence, we predict a write-write data race. This is wrong as both writes are mutually exclusive to each other. Hence, among the guaranteed races reported by \SSHB{},
there may be still some false positives due to inaccurate release-acquire dependencies.

Adapting RoadRunner by following the instrumentation scheme outlined
in Definition~\ref{def:instrumentation} would fix
the issue of inaccurate release-acquire dependencies.
However, this turned out to be difficult and therefore we decided for another option.
We use the original RoadRunner version.
For each guaranteed data race report, we apply a check to filter out false positives
due to inaccurate release-acquire dependencies.
This check is rather simple and is carried out via
the lockset algorithm~\cite{Dinning:1991:DAA:127695:122767}.

The lockset algorithm tries to determine if a variable is always protected by the same mutex.
If it cannot determine at least a single mutex that is always used,
it reports a potential data race for this variable and its accesses.
For this, every thread maintains a set of mutex that it currently owns.
For each acquire event the mutex is added to the lockset of the thread.
For every release event the mutex is removed from the set.

The important insight is that the lockset algorithm
is  is not affected by inaccurate release-acquire dependencies.
This is so because for computing the lockset in each thread,
we only care about inter-thread acquire-release dependencies.
To check if guaranteed races are affected by inaccurate release-acquire dependencies,
we run the lockset algorithm simultaneously with our \SSHB{} algorithm.
We flag a guaranteed data race reported by \SSHB{} as a false positives,
if the intersection of the locksets of the events involved in the race is non-empty.

\begin{figure}
    \bda{ll|l|ccc}
   & \thread{1}{}
   & \thread{2}{}
   & Lockset(1)
   & Lockset(2)
   & Lockset(x) 
  \\ \hline
  1. & \lockE{y} & & \{y\}
  \\ 2. & \writeE{x} & & & & \{y\}
  \\ 3. & & \lockE{y} & & \{y\}
  \\ 4. & & \writeE{x} & & & \{y\}
  \\ 5. & \unlockE{y} & & \emptyset
  \\ 6. & & \unlockE{y} & &  \emptyset
  \eda
  \caption{Inaccurate Trace}
  \label{fig:inaccurate-trace-lockset-ex}
\end{figure}

In Figure~\ref{fig:inaccurate-trace-lockset-ex}
we revisit the earlier example that
explains the problem of inaccurate release-acquire dependencies (Figure~\ref{fig:example-lock-false-positive}).
For this example we find that both threads own mutex $y$ when they access variable $x$.
Hence, the intersection of the locksets of both writes is non-empty.
Hence, lockset does not report a data race as a common mutex for all accesses on $x$ exists.
On the other hand, \SSHB{} reports a guaranteed data race that involves both writes.
Based on the information provided by the lockset algorithm,
we identify this guaranteed race as a false positive.

\begin{table}
%    \small
    \begin{tabular}{c|c|c|c|c|c|c|c}
        \hline
        \textbf  & moldyn  & raytracer & xalan & lusearch & tomcat & avrora & h2  \\ \hline
        \SSHB{}(FP) & 0 & 0 & 0 & 1 & 2 & 6 & 23
    \end{tabular}
    \caption{False positives due to inaccurate release-acquire dependencies.}
    \label{tab:rel-acq-inaccuary-fps}
\end{table}

Table \ref{tab:rel-acq-inaccuary-fps} shows how many  guaranteed data races
reported by \SSHB{} are identified as false positives due to inaccurate release-acquire dependencies.
Benchmarks lusearch, tomcat, avrora and h2 are affected by this tracing inaccuracy.
For example, in case of avrora
six out of 20 reported guaranteed data races have a non-empty lockset
and are therefore false positives.

%%%%%%%%%%%%%%%%%%%%%%%%%%%%%%%%%%%%%%%%%%%%%%%%%%%%%%%%%%%%%%
%%%%%%%%%%%%%%%%%%%%%%%%%%%%%%%%%%%%%%%%%%%%%%%%%%%%%%%%%%%%%%
\section{Related Works}
\label{sec:related-works}

%%%%%%%%%%%%%%%%%%%%%%%%%%%%%%%%%%%%%%%%%%%%%%%%%%%%%%%%%%%%%%
\subsection{Instrumentation and Tracing}

The RoadRunner tool by Flanagan and Freund~\cite{flanagan2010roadrunner}
is used by numerous data race prediction methods.
RoadRunner relies on the ASM code manipulation tool
by Bruneton, Lenglet and Coupaye~\cite{bruneton2002asm}.
As discussed in Section~\ref{sec:roadrunner-results},
acquire and release events are issued \emph{after} the respective operation has executed.
Then, tracing of release-acquire dependencies may be inaccurate
and false positives as described by the example in Figure~\ref{fig:example-lock-false-positive}
may arise.
A further source of inaccurate tracing are unsynchronized reads/writes.
There is no guarantee that their order as found in the trace is accurate
and corresponds to an actual program run.

Marek, Zheng, Ansaloni, Sarimbekov, T{\r{u}}m and Qi~\cite{10.1007/978-3-642-35182-2_18}
present the DiSL framework for instrumentation of Java bytecode.
The objectives of DiSL are similar to RoadRunner.
As we have learnt~\cite{andrea-rosa-at-mpl19}, DiSL could be adapted
to accurately capture release-acquire dependencies. This would require two instrumentation runs.
The first run instruments all instructions that need a post-instrumentation
including the acquire instructions. Release instructions are instrumented separately
in the second run where the instrumentation needs to be added before
the release instruction to capture release-acquire dependencies correctly.
To our knowledge, accurate tracing of write-read dependencies is beyond
the capabilities of DiSL.

Kalhauge and Palsberg~\cite{kalhauge2018sound}
implement their own instrumentation based on ASM instead of using RoadRunner.
This is necessary as for the purpose of deadlock prediction
they need to trace acquire operations that may be potentially being executed.
Recall that RoadRunner only issues a trace event after acquire executes.
%% The issue of unsynchronized reads/writes and trace inaccuracies remain.

RVPredict~\cite{rvpredict} is a commercial tool suite
that supports dynamic data race prediction and comes with
its own tracing tool. We do not know if (and how) the issue
of inaccurate tracing is addressed by RVPredict.

Cao, Zhang, Sengupta, Biswas and Bond~\cite{Cao:2017:HRD:3134419.3108138}
study the issue of efficient tracking of dependencies in a parallel run-time setting.
They rely on a managed run-time, the Jikes R(esearch)VM~\cite{javajikes}, and
employ thread-local traces to reduce the online tracking overhead.
Their method appears to capture accurately write-read dependencies
but this is not further explored as for data race prediction
they only consider FastTrack which ignores write-read dependencies.

%% MS: omit, maybe leads to far
%%
%% \ms{fixme, our own work, tracing of channel-based dependencies,
%%   we are accurate, but, ..., comm links, how could this scale
%%   in the shared memory setting ...}

%% MS: omit
%% Introduce the following notions.
%% Instrumentation atomiticity = atomic update of the objects state.
%% Instrumentation-access atomiticity = actual access and update are atomic.
%% Quote from the paper: ``dynamic data race detection, which requires only instrumentation atomicity, because it does not need to
%% know the order of racy accesses''.
%% But then WRDs might not be accurate. The paper refers only to FastTrack which
%% doesn't care about WRDs.
%% Main point is to ensure that the object's state is updated consistently
%% with as little run-time overhead as possible.
%% It remains unclear if write-read dependencies can be accurately tracked/traced
%% (without additional overhead such as tracking values, incrementing counters).
%% Consider the actual run.
%% \begin{verbatim}
%%     T1      T2      T3
%% 
%% 1.          w(x)
%% 2.                  r(x)
%% 3.  w(x)
%% 
%% They use thread-local log to track dependencies!
%% Need global, atomic counter to guarantee that 
%% w(x)_1 and r(x)_2 are in a write-read dependencies.
%% 
%% Don't show explicitely how to build happens-before edges.
%% \end{verbatim}
%% 
%% \ms{thoughs}
%% Instrumentation-access atomicity not enough.
%% Additionally need to link read to the actual last write,
%% e.g. by using a unique counter or track values.

Hardware-based tracing \cite{DBLP:conf/icse/ShengVEHCZ11} is more accurate and incurs
a much lower run-time overhead compared to the software-based instrumentation schemes
mentioned above. However, as there is only a limited amount of debug registers available to store events,
the program behavior cannot be recorded continuously. Hence, we may miss events and this
means the analysis is less accurate.

Sheng, Vachharajani, Eranian, Hundt, Chen and Zheng~\cite{DBLP:conf/icse/ShengVEHCZ11}
introduce a tracing framework where two tracers are used together.
The first tracer is the hardware tracer to record memory accesses.
The second tracer is a standard software instrumentation to record acquire/release operations
This leads to inaccuracies (release-acquire dependencies) in the trace
but is compensated by lockset algorithm.
Both traces need to be merged for further analysis. The merged trace contains potential
inaccuracies as the exact timing among acquire/release and write/read operations
may be not exact.

%%%%%%%%%%%%%%%%%%%%%%%%%%%%%%%%%%%%%%%%%%%%%%%%%%%%%%%%%%%%%%
\subsection{Happens-Before Based Data Race Prediction}

Mathur, Kini and Viswanathan~\cite{Mathur:2018:HFR:3288538:3276515}
observe that standard happens-before based methods such as FastTrack~\cite{flanagan2010fasttrack} are only sound up to the first race predicted.
Subsequent races reported may be inaccurate (false positives).
Mathur et.~al. propose SHB, a happens-before based data race predictor
that takes into account write-read dependencies and thus avoids
false positives reported by FastTrack.
The accuracy of SHB relies on the accuracy of traces.
Recall the examples from Section~\ref{sec:overview}.

The happens-before based data race prediction methods
by Huang, Meredith and Rosu~\cite{Huang:2014:MSP:2666356.2594315}
and Kalhauge and Palsberg~\cite{kalhauge2018sound}
additionally trace the values read and written to variables
and include control flow events like if-conditions and while-loops (branch-event).
The motivation is to explore more feasible trace reorderings.
Inaccurate results due to inaccurate trace may arise.
Due to space limitations, we refer to Appendix~\ref{sec:value-tracing}.

Biswas, Cao, Zhang, Bond and Wood~\cite{biswas2017lightweight}
use sampling to reduce the run-time overhead.
There is clearly a loss of precision (accuracy) which they are willing
to trade in for improved run-time performance.

Huang and Rajagopalan~\cite{Huang:2016:PMR:3022671.2984024} study the issue of incomplete trace in the data race setting.
They propose a refined analysis method based on an SMT-solver to reduce the amount of false positives due to incomplete trace information.
Programmer annotations are necessary to guide the solver in case of external functions that manipulate shared memory.

Wood, Cao, Bond and Grossman~\cite{DBLP:journals/pacmpl/WoodCBG17}
consider how to reduce the synchronization overhead of online race predictors
such as FastTrack.
Wilcox, Flanagan and Freund~\cite{Wilcox:2018:VEF:3200691.3178514}
give a redesign of the FastTrack method to ensure that
the resulting implementation is free of concurrency bugs.
Both works do not address the issue of inaccurate tracing.

There are numerous other works
that consider trace-based analysis methods to predict concurrency bugs.
For example, consider~\cite{Wang:2011:PCF:2341616.2341619,Said:2011:GDR:1986308.1986334,Wang:2009:SPA:1693345.1693367}.
To the best of our knowledge, the common assumption
is that traces are accurate.

%% MS: omit
%% \ms{possibly move to appendix}
%% In our experiments, we use RoadRunner but include a pre-processing step to
%% fix the trace such that {\bf INST\_RA} is guaranteed.
%% \ms{KAI, need to add further details, there's an unambiguous fix?}
%% 
%% \ms{rough notes}
%% 
%% \begin{verbatim}
%% There's no canonical fix.
%% 
%% Consider
%% 
%% #1         #2
%% 
%% acq(y)
%%            w(x)
%%            acq(y
%% rel(y)
%% w(x)
%% r(x)
%%            rel(y)
%% 
%% 
%% Possible fixes are:
%% 
%% (1)
%% 
%% #1         #2
%% 
%% acq(y)
%% rel(y)
%% w(x)
%% r(x)
%%            w(x)
%%            acq(y)
%%            rel(y)
%% 
%% (2)
%% 
%% #1         #2
%% 
%% acq(y)
%% rel(y)
%% w(x)
%%            w(x)          
%% r(x)
%%            acq(y)
%%            rel(y)
%% 
%% 
%% Note the WRD in case of (2).
%% 
%% Fact:  Can easly check of acq/rel are accurate.
%% 
%% Fact; There's no canonical fix.
%% 
%% Conjecture: Any fix will do for our "accuracy" analysis.
%%             That is, all fixes yield the same result.
%% 
%% \end{verbatim}

%%%%%%%%%%%%%%%%%%%%%%%%%%%%%%%%%%%%%%%%%%%%%%%%%%%%%%%%%%%%%%
%%%%%%%%%%%%%%%%%%%%%%%%%%%%%%%%%%%%%%%%%%%%%%%%%%%%%%%%%%%%%%
\section{Conclusion}
\label{sec:conclusion}

Happens-before based data race prediction methods rely on the accuracy of
the trace of recorded events to guarantee that races are accurate.
Unless some sophisticated and likely costly tracing method is used,
e.g.~consider~\cite{Cao:2017:HRD:3134419.3108138},
write-read dependencies recorded may not be accurate.
We offer diagnostic methods to check
if (1) a race is guaranteed to be correct regardless of 
any potential inaccuracies, (2) maybe incorrect due to inaccurate tracing.
Our experimental results show that our methods adds value to existing
data race predictors such as FastTrack and SHB.
In future work,
we plan to identify undetected races. Recall the example
in Figure~\ref{fig:example-wrd-false-negative} where a race may be unnoticed
due to inaccurate tracing.

\bibliography{main}

\newpage

\appendix

\emph{Optional material}

%%%%%%%%%%%%%%%%%%%%%%%%%%%%%%%%%%%%%%%%%%%%%%%%%%%%%%%%%%%%%%
%%%%%%%%%%%%%%%%%%%%%%%%%%%%%%%%%%%%%%%%%%%%%%%%%%%%%%%%%%%%%%
\section{Proofs}
\label{sec:proofs}

%%%%%%%%%%%%%%%%%%%%%%%%%%%%%%%%%%%%%%%%%%%%%%%%%%%%%%%%%%%%%%
\subsection{Auxiliary Results}
\label{sec:auxiliary}

\begin{proposition}
  \label{prop:w1-unsync}
  Let $T$ be a trace and $e$ be some read event.
  Let $f, g \in W_1(e)$ be two distinct write events.
  Then, $f$ and $g$ are unsynchronized
  with respect to each other.
\end{proposition}
\begin{proof}
  Suppose $f, g \in W_1(e)$.
  By definition, there is no other write in between.
  Hence, $f$ and $g$ are unsynchronized with respect to each other.
  \qed
\end{proof}  

\begin{proposition}
    \label{prop:w2-unsync}
  Let $T$ be a trace and $e$ be some read event.
  Let $f, g \in W_2(e)$ be two distinct write events.
  Then, $f$ and $g$ are unsynchronized
  with respect to each other.
\end{proposition}
\begin{proof}
  Suppose $f, g \in W_2(e)$.
  By definition, there is no other write in between.
  Hence, $f$ and $g$ are unsynchronized with respect to each other.
  \qed
\end{proof}  

Consider $f \in W_1(e)$ and $g \in W_2(e)$.
Assuming $\hb{f}{g}$ immediately leads to a contradiction.
On the other hand, $\hb{g}{f}$ is possible as shown by
the following example.

\begin{example}
  Consider the following trace.
  \bda{ll|l|l}
  & \thread{1}{} & \thread{2}{} & \thread{3}{}
  \\ \hline
  1. & \lockE{y_1} &&
  \\
  2. & \lockE{y_2} &&
  \\
  3. & \writeE{x} &&
  \\
  4. & \unlockE{y_2} &&
  \\
  5. & \unlockE{y_1} &&
  \\
  6. & & & \lockE{y_2}
  \\
  7. & & & \writeE{x}
  \\
  8. & & & \unlockE{y_2}
  \\
  9. & & \lockE{y_1} &
  \\
  10. & & \readE{x} &
  \\
  11. & & \unlockE{y_1} &
  \eda
  Consider the read event $\readEE{x}{10}$.
  We find $W_1 = \{ \writeEE{x}{7} \}$
  and $W_2 = \{ \writeEE{x}{3} \}$
  where $\hb{\writeEE{x}{3}}{\writeEE{x}{7}}$.
\end{example}

%%%%%%%%%%%%%%%%%%%%%%%%%%%%%%%%%%%%%%%%%%%%%%%%%%%%%%%%%%%%%%
\subsection{Proof of Proposition~\ref{prop:prop}}

\begin{proof}
  If the last write event $w$ as characterized by Definition~\ref{def:schedulable-happens-before}
  is unsynchronized w.r.t.~the read $e$, then
  $w \in W_1(e)$.
  Otherwise, $w$ and $e$ are in happens-before relation.
  That is, $\hb{w}{e}$.
  By the conditions in Definition~\ref{def:schedulable-happens-before}
  there is no other write (on the same variable) in between.
  Hence, $w \in W_2(e)$.
  \qed
\end{proof}

%%%%%%%%%%%%%%%%%%%%%%%%%%%%%%%%%%%%%%%%%%%%%%%%%%%%%%%%%%%%%%
\subsection{Proof of Proposition~\ref{prop:prop-one}}

\begin{proof}
  For each of the writes in one of the sets $W_1(e)$
  and $W_2(e)$ applies the condition that
  there is no other write with the same property (happening) in between.
  Therefore, writes in $W_1(e)$ are unsynchronized
  with respect to the other.
  The same applies to $W_2(e)$.
  See Propositions~\ref{prop:w1-unsync} and~\ref{prop:w2-unsync}.
  Hence, there cannot be more than $O(k)$ candidates to consider.
  
  For both sets the worst-case number of candidates may arise
  as shown by the following example.

   \bda{ll|l|l|l|l|l}
  & \thread{1}{} & \thread{2}{} & \thread{3}{} & \thread{4}{} & \thread{5}{}
  \\ \hline
  1. & && &&  \writeE{x}
  \\
  2. & && & \writeE{x} &
  \\
  3. && & \lockE{y_1} &&
  \\
  4. && & \writeE{x} &&
  \\
  5. && & \unlockE{y_1} &&
  \\
  6. &&  \lockE{y_2} && &&
  \\
  7. &&  \writeE{x} && &&
  \\
  8. &&  \unlockE{y_2} && &&
  \\
  9. & \lockE{y_1} & && &&
  \\
  10. & \lockE{y_2} & && &&
  \\
  11. & \unlockE{y_2} & && &&
  \\
  12. & \unlockE{y_1} & && &&
  \\
  13. & \readE{x} & && &&
  \eda
  We have that $W_1(\readEE{x}{13}) = \{ \writeEE{x}{1}, \writeEE{x}{2} \}$
  and $W_2(\readEE{x}{13}) = \{ \writeEE{x}{4}, \writeEE{x}{7} \}$
  where $k=5$.
  It is easy to generalize the example for arbitrary $k$.
  \qed
\end{proof}

%%%%%%%%%%%%%%%%%%%%%%%%%%%%%%%%%%%%%%%%%%%%%%%%%%%%%%%%%%%%%%
\subsection{Proof of Proposition~\ref{prop:prop-two}}

\begin{proof}
  We prove the statement by induction over the events in the trace.

  Consider the case of a read event $e$.
  The write event chosen $f$ by WRD in Definition~\ref{def:schedulable-happens-before}
  is clearly an element in $W(e)$.
  If $f \in W_2(e)$ we are immediately done.

  Consider the case that $f \in W_1(e)$.
  We need to show that by imposing $\shb{f}{e}$ we obtain a partial order.
  The only reason why this would not be the case is that there is
  some cyclic dependency. However, this is impossible as $\hb{}{}$
  and $\shb{}{}$ are built up by only considering events that
  appear earlier in the trace.
  That is, for two events $g_1$ and $g_2$ in $\hb{}{}$ and $\shb{}{}$ relation
  we conclude that $\pos{g_1} < \pos{g_2}$ and therefore
  cycles are impossible.
  \qed
\end{proof}

%%%%%%%%%%%%%%%%%%%%%%%%%%%%%%%%%%%%%%%%%%%%%%%%%%%%%%%%%%%%%%
\subsection{Proof of Proposition~\ref{prop:someshb-hb-equivalence}}

\begin{proof}
  For each read event $e$, we pick a candidate $f \in W_2(e)$.
  By assumption such a candidate $f$ always exists and
  enjoys the property $\hb{e}{f}$. Hence,
  the relation $\someshb{f}{e}$ can already be derived via the HB relation.
  \qed
\end{proof}

%%%%%%%%%%%%%%%%%%%%%%%%%%%%%%%%%%%%%%%%%%%%%%%%%%%%%%%%%%%%%%
\subsection{Proof of Proposition~\ref{prop:guaranteed-data-race}}

\begin{proof}
  Assume the contrary. Then, there would be a path. Contradiction.
  \qed
\end{proof}

%%%%%%%%%%%%%%%%%%%%%%%%%%%%%%%%%%%%%%%%%%%%%%%%%%%%%%%%%%%%%%
\subsection{Proof of Proposition~\ref{prop:maybe-data-race}}

\begin{proof}
  We use the following notation.
  We will write $w, w'$ for write events and $r, r'$ for read events.

  We assume a path exists. Nodes along the path are distinct.
  Hence, he label $W(r)$ appears at most once and we have
  that for each set $W(r)$ of WRD candidates we consider a specific choice.
  
  We need to show that based on our choice of WRD candidates to built the path
  between $e$ and $f$, we can built a $\someshb{}{}$ relation that is a partial order.
  So, we need to consider the read events for which no WRD candidate has been chosen so far.
  We will show that there for the candidate-free read events $r$ we can choose
  a write candidate $w \in W_2(r)$. The resulting $\someshb{}{}$ relation
  is cycle-free and thus forms a proper partial order.

   We consider the path where w.l.o.g.~the path starts with $e$.
  \bda{c}
   e \rightarrow \underbrace{g_1 \rightarrow \dots \rightarrow g_n}_{\dots \rightarrow r' \rightarrow \dots w \rightarrow r \rightarrow \dots} \rightarrow f
  \eda
  We consider the situation that a specific candidate $w \in W(r)$ has been chosen to built this path.
  For the moment, we assume that $w \rightarrow r$ is the only non-HB edge (1).
  We further assume $w \in W_1(r)$.
  Otherwise, $w \in W_2(r)$. Then, Proposition~\ref{prop:someshb-hb-equivalence}
  applies and we are immediately done.
  Hence, $w \in W_1(r)$ which implies that $w$ and $r$ are not in a $\hb{}{}$ relation (2).
  
  We consider a candidate-free read element $r'$ where we pick $w' \in W_2(r')$.
  If there is no such $r'$ we are immediately done.
  Suppose that adding the edge $w' \rightarrow r'$ creates a cycle.
  This means that $w'$ must be part of the cycle. In particular, we conclude that
  \bda{c}
   w' \rightarrow r' \rightarrow \dots w \rightarrow r' \rightarrow w'
  \eda
  The case that $w'$ appears before $w$ immediately leads to a contradiction as
  we assume that $\hb{}{}$ is a partial order and therefore cycle-free.

  From above, we conclude that there is a HB path that contains a cycle.
  Recall our assumption~(1). That is, $w' \rightarrow \dots \rightarrow w'$.
  This is clearly a contradiction.
  Hence, adding the edge $w' \rightarrow r'$ does not create a cycle.

  We yet need to consider the situation that the candidate-free $r'$
  appears after $r$. As only HB edges are involved, we immediately
  conclude that adding the edge $w' \rightarrow r'$ does not create a cycle.

  What remains is to generalize assumption~(1).
  In case of multiple non-HB edges $w_i \rightarrow r_i$ we proceed by induction
  over the number of such cases by applying the above reasoning steps.
  \qed
\end{proof}

%%%%%%%%%%%%%%%%%%%%%%%%%%%%%%%%%%%%%%%%%%%%%%%%%%%%%%%%%%%%%%
%%%%%%%%%%%%%%%%%%%%%%%%%%%%%%%%%%%%%%%%%%%%%%%%%%%%%%%%%%%%%%
\section{FastTrack}

\begin{algorithm}
    \caption{Vector Clocks for Predicting FastTrack Races (FastTrack\ algorithm)}\label{alg:ft-race-checker}
    
    \begin{tabular}{ll}
        \begin{minipage}{.5\textwidth}  
            \begin{algorithmic}[1]
                \Procedure{acquire}{$i,x$}
                \State $\threadVC{i} = \supVC{\threadVC{i}}{\lockVC{x}}$
                \EndProcedure
            \end{algorithmic}
            
            \begin{algorithmic}[1]
                \Procedure{write}{$i,x$}
                \For {$\thread{j}{k} \in \readVC{x}$}
                \If {$k > \accVC{\threadVC{i}}{j}$}
                \State $races = races \cup \{\threadVC{i}\}$
                \EndIf
                \EndFor
                \If { $\thread{j}{k} = \lastWriteVC{x} \wedge k > \accVC{\threadVC{i}}{j} $}
                \State $races = races \cup \{\threadVC{i}\}$
                \EndIf
                \State $\lastWriteVC{x} = \accVC{\threadVC{i}}{i}$
                \State $\incC{\threadVC{i}}{i}$
                \EndProcedure
            \end{algorithmic}

        \end{minipage}
        
        &
        
        \begin{minipage}{.5\textwidth}  
            
            \begin{algorithmic}[1]
                \Procedure{release}{$i,x$}
                \State $\lockVC{x} = \threadVC{i}$  
                \State $\incC{\threadVC{i}}{i}$
                \EndProcedure
            \end{algorithmic}
            
            \begin{algorithmic}[1]
                \Procedure{read}{$i,x$}
                \If { $\thread{j}{k} = \lastWriteVC{x} \wedge k > \accVC{\threadVC{i}}{j} $}
                \State $races = races \cup \{\threadVC{i}\}$
                \EndIf
                \State $\readVC{x} =  \{\thread{i}{\accVC{\threadVC{i}}{i}}\} \cup\{\thread{j}{k} \mid \thread{j}{k} \in \readVC{x} \wedge k > \accVC{\threadVC{i}}{j}\}$ 
                \State $\updateVC{\readVC{x}}{i}{\accVC{\threadVC{i}}{i}}$
                \State $\incC{\threadVC{i}}{i}$
                \EndProcedure
            \end{algorithmic}
            
        \end{minipage}
        
    \end{tabular}
    
\end{algorithm}

 Algorithm \ref{alg:ft-race-checker} shows an overview of the FastTrack algorithm from \cite{flanagan2010fasttrack}.

For each thread $i$ we maintain a vector clock $\threadVC{i}$. For each variable $x$, we find a epoch $\lastWriteVC{x}$
to maintain the last write access to $x$. Similarly, we find $\readVC{x}$ to maintain the concurrent read accesses.
For each mutex $y$, we find vector clock $\lockVC{y}$ to maintain the last release event on $y$.

Initially, for each vector clock $\threadVC{i}$
all time stamps are set to 0 but position $i$ where the time stamp is set to 1.
For $\readVC{x}$, $\lastWriteVC{x}$ and $\lockVC{y}$ all time stamps are set to 0.
We write $\incC{V}{i}$ as a short-hand for incrementing the vector clock $V$ at position $i$ by one.

We consider the various cases of Algorithm \ref{alg:ft-race-checker}. For acquire and release events, parameter $i$ refers to the thread id and $x$ refers to the name of the mutex. Similarly, for writes and reads, $i$ refers to the thread id and x to the name of the variable.

In case of an acquire event we synchronize the thread's vector clock with the most recent release event by building the union of the vector clocks $\threadVC{i}$ and $\lockVC{x}$. In case of a release event, we update $\threadVC{x}$.

In case of a write event we compare the thread's vector clock against the read epochs in $\readVC{x}$ and the single write epoch stored in $\lastWriteVC{x}$ to check for a write-read and write-write race. Then we update $\lastWriteVC{x}$ to record the epoch of the most recent write on $x$. Last step is to increment the thread vector clock at position $i$ by one.

In case of a read event, we check for read-write races by comparing the thread's vector clock against the last write epoch. The next step is to update the set of read epochs to contain only the concurrent reads.

%%%%%%%%%%%%%%%%%%%%%%%%%%%%%%%%%%%%%%%%%%%%%%%%%%%%%%%%%%%%%%
%%%%%%%%%%%%%%%%%%%%%%%%%%%%%%%%%%%%%%%%%%%%%%%%%%%%%%%%%%%%%%
\section{\HBPartner{} and \SHBPartner{} Algorithms}

Variants of FastTrack and SHB
that report pairs of events that are in a race.
See Algorithms~\ref{alg:vcHB-race-pairs} and~\ref{alg:vcSHB-race-pairs}.

\begin{algorithm}
    \caption{Vector Clocks for Predicting \HBPartner\ Race-Pairs (\HBPartner\ algorithm)}\label{alg:vcHB-race-pairs}
    
    \begin{tabular}{ll}
        \begin{minipage}{.8\textwidth}
            \begin{algorithmic}[1]
                \Procedure{acquire}{$i,x$}
                \State $\threadVC{i} = \supVC{\threadVC{i}}{\lockVC{x}}$
                \EndProcedure
            \end{algorithmic}
            
            \begin{algorithmic}[1]
                \Procedure{write}{$i,x, k$}
                \For{$\thread{j}{k} \in cw(x)$}
                \If {$k > \accVC{\threadVC{i}}{j}$}
                \State $races = races \cup \{(\thread{j}{k}, \thread{i}{\accVC{\threadVC{i}}{i}})\}$
                \EndIf
                \EndFor
                \State $cw(x) = \{ \thread{i}{\accVC{\threadVC{i}}{i}} \}
                \cup \{ \thread{j}{k} \mid \thread{j}{k} \in cw(x) \wedge
                k > \accVC{\threadVC{i}}{j} \}$
                
                \For{$\thread{j}{k} \in cr(x)$}
                \If {$k > \accVC{\threadVC{i}}{j}$}
                \State $races = races \cup \{(\thread{j}{k}, \thread{i}{\accVC{\threadVC{i}}{i}})\}$
                \EndIf
                \EndFor
                \State $\incC{\threadVC{i}}{i}$
                \EndProcedure
            \end{algorithmic}

            \begin{algorithmic}[1]
                \Procedure{release}{$i,x$}
                \State $\lockVC{x} = \threadVC{i}$  
                \State $\incC{\threadVC{i}}{i}$
                \EndProcedure
            \end{algorithmic}
            
            \begin{algorithmic}[1]
                \Procedure{read}{$i,x, k$}
                \For{$\thread{j}{k} \in cw(x)$}
                \If {$k > \accVC{\threadVC{i}}{j}$}
                \State $races = races \cup \{(\thread{i}{\accVC{\threadVC{i}}{i}}, \thread{j}{k})\}$
                \EndIf
                \EndFor
                \State $cr(x) = \{ \thread{i}{\accVC{\threadVC{i}}{i}} \}
                \cup \{ \thread{j}{k} \mid \thread{j}{k} \in cr(x) \wedge
                k > \accVC{\threadVC{i}}{j} \}$               
                \State $\incC{\threadVC{i}}{i}$
                \EndProcedure
            \end{algorithmic}
            
        \end{minipage}
        
    \end{tabular}
    
\end{algorithm}

\begin{algorithm}
    \caption{Vector Clocks for Predicting \SHBPartner\ Race-Pairs (\SHBPartner\ algorithm)}\label{alg:vcSHB-race-pairs}
    
    \begin{tabular}{ll}
        \begin{minipage}{.8\textwidth}
            \begin{algorithmic}[1]
                \Procedure{acquire}{$i,x$}
                \State $\threadVC{i} = \supVC{\threadVC{i}}{\lockVC{x}}$
                \EndProcedure
            \end{algorithmic}
            
            \begin{algorithmic}[1]
                \Procedure{write}{$i,x, k$}
                \For{$\thread{j}{k} \in cw(x)$}
                \If {$k > \accVC{\threadVC{i}}{j}$}
                \State $races = races \cup \{(\thread{j}{k}, \thread{i}{\accVC{\threadVC{i}}{i}})\}$
                \EndIf
                \EndFor
                \State $cw(x) = \{ \thread{i}{\accVC{\threadVC{i}}{i}} \}
                \cup \{ \thread{j}{k} \mid \thread{j}{k} \in cw(x) \wedge
                k > \accVC{\threadVC{i}}{j} \}$
                
                \For{$\thread{j}{k} \in cr(x)$}
                \If {$k > \accVC{\threadVC{i}}{j}$}
                \State $races = races \cup \{(\thread{j}{k}, \thread{i}{\accVC{\threadVC{i}}{i}})\}$
                \EndIf
                \EndFor
                \State $\lastWriteVC{x} = \threadVC{i}$
                \State $\incC{\threadVC{i}}{i}$
                \EndProcedure
            \end{algorithmic}

            \begin{algorithmic}[1]
                \Procedure{release}{$i,x$}
                \State $\lockVC{x} = \threadVC{i}$  
                \State $\incC{\threadVC{i}}{i}$
                \EndProcedure
            \end{algorithmic}
            
            \begin{algorithmic}[1]
                \Procedure{read}{$i,x, k$}
                \For{$\thread{j}{k} \in cw(x)$}
                \If {$k > \accVC{\threadVC{i}}{j}$}
                \State $races = races \cup \{(\thread{i}{\accVC{\threadVC{i}}{i}}, \thread{j}{k})\}$
                \EndIf
                \EndFor
                \State $\threadVC{i} = \supVC{\threadVC{i}}{\lastWriteVC{x}}$
                \State $cr(x) = \{ \thread{i}{\accVC{\threadVC{i}}{i}} \}
                \cup \{ \thread{j}{k} \mid \thread{j}{k} \in cr(x) \wedge
                k > \accVC{\threadVC{i}}{j} \}$               
                \State $\incC{\threadVC{i}}{i}$
                \EndProcedure
            \end{algorithmic}
            
        \end{minipage}
        
    \end{tabular}
    
\end{algorithm}

%%%%%%%%%%%%%%%%%%%%%%%%%%%%%%%%%%%%%%%%%%%%%%%%%%%%%%%%%%%%%%
%%%%%%%%%%%%%%%%%%%%%%%%%%%%%%%%%%%%%%%%%%%%%%%%%%%%%%%%%%%%%%
\section{Missed data races due to unique source code locations}

Algorithms FastTrack, SHB and our variants ignore races that result from the same code location.
Race pairs connected to the same code locations may be affected by different write-read dependencies.
For example, consider conditional statements inside a loop.
As our diagnostic phase does not distinguish between such race pairs and only consider race pairs based on
their source location, we may miss some guaranteed data races as the following example shows.

For the following example code and trace the guaranteed data race on $x$ would be ignored. The first occurrence of the data race between the source code locations six in $Program \thread{1}{}$ and five in $Program \thread{2}{}$ is a false positive. The write read dependency on $y$ introduces a happens-before relation between the accesses on $x$.
The second occurrence is a guaranteed data race since it does not have a write read dependency on $y$. The write read dependency only occurs for the first loop execution where $z = 0$ applies. Repeated data races like the data race on $x$ are ignored in our benchmarks to reduce the amount of data races that need to be checked. Since the second occurrence is ignored, our algorithm will not report a data race for the writes on $x$. Testing all occurrences of data races would solve the problem but increases the time necessary to test the program.

\begin{minipage}{\linewidth}
    \begin{tabular}{l|l|l}
         Program $\thread{1}{}$ & Program $\thread{2}{}$ & Trace \\ \hline
         \begin{minipage}{.3\linewidth}
             \begin{algorithmic}[1]
                 \State z = 0
                 \For{}
                 \If {z == 0}
                 \State y = 1
                 \EndIf
                 \State x = 4
                 \State z = 1
                 \EndFor
             \end{algorithmic}
         \end{minipage}
     &
     \begin{minipage}{.3\linewidth}
         \begin{algorithmic}[1]
             \For{}
             \If {z == 0}
             \State tmp = y
             \EndIf
             \State x = 3
             \EndFor
         \end{algorithmic}
     \end{minipage}
    &
    \begin{minipage}{.3\linewidth}
        \begin{tabular}{ll|l}
            \textbf{} & $\thread{1}{}$ & $\thread{2}{}$ \\ \hline
            1. & $r(z)$ \\
            2. & $w(y)$ \\
            3. & $w(x)$ \\
            4. & & $r(z)$ \\
            5. && $r(y)$ \\
            6. && $w(x)$ \\
            7. &w(z) \\
            8. &$r(z)$ \\
            9. & $w(x)$ \\
            10. && $r(z)$ \\
            11. & $w(x)$ 
        \end{tabular}
    \end{minipage}
    \end{tabular}
\end{minipage}

%%%%%%%%%%%%%%%%%%%%%%%%%%%%%%%%%%%%%%%%%%%%%%%%%%%%%%%%%%%%%%
%%%%%%%%%%%%%%%%%%%%%%%%%%%%%%%%%%%%%%%%%%%%%%%%%%%%%%%%%%%%%%
\section{Value tracing}
\label{sec:value-tracing}

The works of \cite{Huang:2016:PMR:3022671.2984024} and \cite{Kalhauge:2018:SDP:3288538.3276516} trace the values that are read and written for each read and write event.
These extended traces potentially contain inaccuracies that influence the result of the race prediction.
Consider the following two examples.

\begin{example}
    
In the following example we find two possible write-read dependencies that can not be distinguished by the value written and read on $x$. A data race is reported if the write-read dependency between step three and four is assumed. For the write-read dependency between step two and four, the accesses on $y$ are not in a race.

\begin{figure}
    \bda{ll|l|l}
        &\thread{1}{}  & \thread{2}{} & \thread{3}{} \\ \hline
        1. & y = 42 & \textbf{} & \textbf{} \\
        2. & x = 5 & \textbf{} & \textbf{} \\
        3. & &        x = 5 \\
        4. & & &                 if (x == 5) \\
        5. & & &                 y = 43
    \eda
\end{figure}

\end{example}

\begin{example}
Figure \ref{fig:codeEx} shows how the code could be instrumented to obtain the values. We write $hash(x)$ for the evaluation of the value of $x$. The instrumentation introduces function calls ($trace_w$ and $trace_r$) that trace the write and read events and the associated values. 

The function calls and the write/read events are not atomic, thus arbitrary delays between the actual write/read and the call to the tracing function can occur. In the following example the call $trace_r(x, hash(x))$ is delayed. Because of this delay, another thread is able to update the value of $x$ between the actual read and the function call. 

\begin{figure}
    
    \bda{ll|l|l}
    & \thread{1}{}  & \thread{2}{} & \thread{3}{} \\ \hline
    1. & y = 42 & \textbf{} & \textbf{} \\
    2. & trace_w(y, hash(y)) & \textbf{} & \textbf{} \\
    3. & x = 5 & \textbf{} & \textbf{} \\
    4. & trace_w(x, hash(x)) & \textbf{} & \textbf{} \\
    5. & & &                 if (x == 5) \\
    6. & &        x = 6 \\
    7. & &        trace_w(x, hash(x))\\
    8. & & &                 trace_r(x, hash(x))\\
    9. & & &                 y = 43 \\
    10. & & &                 trace_w(y, hash(y))
    \eda
    \caption{Program execution with code.}
    \label{fig:codeEx}
\end{figure}   
    The delayed $trace_r(x hash(x))$ call leads to the trace in Figure \ref{fig:traceForCodeEx}. According to the trace, the $\thread{3}{\readE{x, 6})}$ event reads the value six which is written by thread two at step three. A race prediction algorithm has to assume the write-read dependency between step three and four due to the traced values. Thus, a data race between step one and five is predicted.
    
    In the actual execution, shown in Figure \ref{fig:codeEx}, the read event reads the value written by thread one. With the write-read dependency that occurred during the program execution, step one and five are not in a race.
    
    Because of this inaccurate trace a false positive for the two accesses on $y$ is reported. 

\begin{figure}
       
    \bda{ll|l|l}
    & \thread{1}{}  & \thread{2}{} & \thread{3}{} \\ \hline
    1. & \writeE{y, 42} & \textbf{} & \textbf{} \\
    2. & \writeE{x, 5} & \textbf{} & \textbf{} \\
    3. & &        \writeE{x,6} \\
    4. & &  &      \readE{x,6}\\
    5. & & &                 \writeE{y, 43} \\
    \eda
    \caption{Trace for the code execution shown in Figure \ref{fig:codeEx}}
    \label{fig:traceForCodeEx}
\end{figure}

This is not limited to false positives. Figure \ref{fig:codeEx2} is similar to the previous example. The only change is that the write on $y$ is performed by thread two in this example. 

Due to this change the data race on $y$ is not detected. A race prediction algorithm that only considers the write-read dependency, will always assume the write-read dependency between step six and three due to the value and discard the data race between step four and nine.

\begin{figure}
    
    \bda{ll|l|l}
    & \thread{1}{}  & \thread{2}{} & \thread{3}{} \\ \hline
    1. & x = 5 & \textbf{} & \textbf{} \\
    2. & trace_w(x, hash(x)) & \textbf{} & \textbf{} \\
    4. & & y = 42  & \textbf{} \\
    5. & & trace_w(y, hash(y))  & \textbf{} \\
    3. & & &                 if (x == 5) \\
    6. & &        x = 6 \\
    7. & &        trace_w(x, hash(x))\\
    8. & & &                 trace_r(x, hash(x))\\
    9. & & &                 y = 43 \\
    10. & & &                 trace_w(y, hash(y))
    \eda
    \caption{Program execution with code.}
    \label{fig:codeEx2}
\end{figure}
\end{example}

%%%%%%%%%%%%%%%%%%%%%%%%%%%%%%%%%%%%%%%%%%%%%%%%%%%%%%%%%%%%%%
%%%%%%%%%%%%%%%%%%%%%%%%%%%%%%%%%%%%%%%%%%%%%%%%%%%%%%%%%%%%%%
\section{Tracing Error With RoadRunner}
\label{sec:roadrunner-error}

We recorded a run of the Tomcat test and additionally enumerated every lock and access event with a global access number to record the order in which they occurred. In the trace in Figure \ref{fig:inaccurate-rel-acq-tomcat}, the entry `1429007:1,WR,2657, JkMain.java:339:237,1423414'  is the 1429007. trace event, performed by thread one and it is a write event on variable 2657 at code location `Threadpool.java:661:12'. The access has 1423414 as its global access number that describes the order in which every memory access occurred. 

In this example  thread 24 acquires mutex 633 (2773. lock operation). 
 It gets stuck trying to read variable 54437 as 
 thread 1 first needs to write on variable 2657 
 (global access number for $\thread{1}{\writeE{2657}} < \thread{24}{\readE{54437}}$. 
 Thread 1 tries to acquire mutex 633 that thread 24
 still owns and therefore is stuck too and cannot reach the write event on variable 2657 that is necessary for thread 24 to continue. Thus we have a cyclic dependency and are `stuck' during our trace replay.
 Ignoring the order in which the memory access occurred leads to inaccurate write-read dependencies.

\begin{figure}
\begin{tabular}{l}
 1428960:24,LK,633,nil,2773\\
 1429001:1,LK,633,nil,2776 STUCK\\
 1429004:1,UK,633,nil,0 \\
 1429007:1,WR,2657,JkMain.java:339:237,1423414 \\
 1429008:24,RD,54437,ThreadPool.java:661:12,1423415 STUCK  \\
 1429012:24,UK,633,nil,0
 \end{tabular}
 \caption{Inaccurate release-acquire dependencies in tomcat test.}
 \label{fig:inaccurate-rel-acq-tomcat}
\end{figure}

\end{document}